 \documentclass[onecolumn, draftclsnofoot,11pt]{IEEEtran}
\usepackage{acronym}
\newacro{CES}[CES]{Cover-El Gamal-Salehi}
\newacro{GKW}[GKW]{G\' acs-K\"orner-Witsenhausen}
\newacro{MAC}[MAC]{Multiple-Access Channel}
\newacro{MAC-FB}[MAC-FB]{MAC with Feedback}
\newacro{CL}[CL]{Cover-Leung}
\newacro{ptp}[PtP]{Point-to-Point}

\usepackage{mathrsfs}
\usepackage{tikz}
\usepackage[utf8]{inputenc}
\usepackage{amsfonts} 
\usepackage{epstopdf}
\usepackage{graphicx}
\usepackage{bbm}
\usepackage{enumerate}
\usepackage{epsf,epsfig,verbatim,amssymb,amsmath,array,cite,amsthm,subfigure,multicol,multirow}  
\usepackage{psfrag,bm,xspace}
\usepackage{hhline}
\usepackage{mathabx}

\newtheorem{theorem}{Theorem}

\newtheorem{lem}{Lemma}

\newtheorem{fact}{Fact}

\newtheorem{proposition}{Proposition}

\theoremstyle{definition}
\newtheorem{definition}{Definition}
\newtheorem{example}{Example}

\newtheorem*{prob*}{Problem}

\theoremstyle{remark}
\newtheorem{remark}{Remark}

\allowdisplaybreaks 
\makeatletter
\def\old@comma{,}
\catcode`\,=13
\def,{%
  \ifmmode%
    \old@comma\discretionary{}{}{}%
  \else%
    \old@comma%
  \fi%
}
\makeatother

\usepackage{xcolor}
\definecolor{darkblue}{rgb}{0.1,0.1,0.8}
\definecolor{brickred}{rgb}{0.8, 0.25, 0.33}
\usepackage{etoolbox}

\providetoggle{Blue_revision}
\settoggle{Blue_revision}{true}
\newcommand{\revcol}{\iftoggle{Blue_revision}{\color{darkblue}}{ }}
\newcommand{\add}[1]{\iftoggle{Blue_revision}{{\color{darkblue}#1}}{#1}}

 \usepackage{hyperref}

\usepackage{graphicx}

\global\long\def\ZZ{\mathbb{Z}}
\global\long\def\EE{\mathbb{E}}
\global\long\def\PP{\mathbb{P}}
\global\long\def\FF{\mathbb{F}}
\global\long\def\11{\mathbbm{1}}

\newcommand{\bfb}{\mathbf{b}}
\newcommand{\bft}{\mathbf{t}}
\newcommand{\bfs}{\mathbf{s}}
\newcommand{\bfu}{\mathbf{u}}
\newcommand{\bfv}{\mathbf{v}}
\newcommand{\bfw}{\mathbf{w}}

\newcommand{\bfB}{\mathbf{B}}
\newcommand{\bfT}{\mathbf{T}}
\newcommand{\bfU}{\mathbf{U}}
\newcommand{\bfV}{\mathbf{V}}

\newcommand{\bfx}{\mathbf{x}}
\newcommand{\bfX}{\mathbf{X}}
\newcommand{\bfy}{\mathbf{y}}
\newcommand{\bfY}{\mathbf{Y}}

\global\long\def\+{\oplus}
\global\long\def\P{\mathsf{P}}

\def\<{\langle}
\def\>{\rangle}

\newcommand*{\medcap}{\mathbin{\scalebox{1}{\ensuremath{\bigcap}}}}%
\newcommand*{\medcup}{\mathbin{\scalebox{1}{\ensuremath{\bigcup}}}}%

\usepackage{stackengine}
\def\deq{\mathrel{\ensurestackMath{\stackon[1pt]{=}{\scriptstyle\Delta}}}}

\newcommand{\sourceS}{(\underline{S},  P_{\underline{S}})}
\newcommand{\MAC}{(\underline{\mathcal{X}}, \mathcal{Y}, P_{Y|X_1X_2X_3} )}
\newcommand{\MACFB}{(\underline{\mathcal{X}}, \mathcal{Y}, P_{Y|\underline{X}}, \mathcal{T})}

\settoggle{Blue_revision}{false}

\begin{document}
\title{\huge{Structured Mappings and Conferencing Common Information for  Multiple-access Channels }}
\author{\IEEEauthorblockN{Mohsen Heidari, S.  Sandeep Pradhan
\thanks{This work was presented in part at IEEE International Symposium on Information Theory (ISIT), July 2016  and July 2017. }}\\
\IEEEauthorblockA{Department of Electrical Engineering and Computer Science,\\
University of Michigan, Ann Arbor, MI 48109, USA.\\
Email: \tt\small \href{mailto:mohsenhd@umich.edu}{mohsenhd@umich.edu}, \href{mailto:pradhanv@umich.edu}{pradhanv@umich.edu}}\\ \today{}}

%
%

%
%




\maketitle
\begin{abstract}
In this work, we study two problems: three-user \ac{MAC} with correlated sources, and \ac{MAC-FB} with independent messages.  
For the first problem, we identify a structure in the joint probability distribution of discrete memoryless sources, and define a new common information called ``conferencing common information". We develop a multi-user joint-source channel coding methodology based on structured mappings to encode this common information efficiently and to transmit it over a \ac{MAC}.  
We derive a new set of sufficient conditions for this coding strategy using single-letter information quantities for arbitrary sources and channel distributions. Next, we make a fundamental connection between this problem and the problem of communication of independent messages over three-user \ac{MAC-FB}. In the latter problem,
although the messages are independent to begin with,
they become progressively correlated given the channel output feedback.
Subsequent communication can be modeled as transmission of 
correlated sources over \ac{MAC}.
Exploiting this connection,  we develop a new coding scheme for the problem.  We characterize its performance using single-letter information quantities, and derive an inner bound to the capacity region. 
For both problems, we provide a set of examples 
where these rate regions are shown to be optimal. Moreover, 
we analytically prove that this performance is not achievable using random unstructured random mappings/codes.


\end{abstract}


\section{Introduction}
Many coding strategies for processing/transmitting sources of information in a distributed fashion harness structures in the statistical description of the sources.  Common information/randomness can be viewed as an example of such a  structure. Efforts in finding a measure of common information among distributed sources led to several definitions \cite{Gacs_Korner_Comm_info,Witsenhausen_comm_info,Wyner_comm_info,ElGamal_comm_info}.  A noteworthy definition of common information is due to G\'acs-K\"orner  \cite{Gacs_Korner_Comm_info} and Witsenhausen \cite{Witsenhausen_comm_info}, which is an information-theoretic measure of the amount of common randomness that can be extracted from two sources. 
\ac{GKW} \textit{common part} between two correlated memoryless sources $(S_1,S_2)$ is defined as a random variable $W$ with the largest entropy, for which there exist functions $f,g$ such that $W=f(S_1)=g(S_2)$ with probability one. 
The random variable $f(S_1)$ (or equivalently $g(S_2)$) represents the ``common randomness" generated from the sources, and the functions $(f, g)$ represent the extraction process applied on the sources. 

\ac{GKW} common part has been found to useful in many problems such as transmission of distributed sources over channels \cite{Gray_Wyner_1974,Wagner_comm_dist_scr,CES,Liu_interference_scr} and distributed key generation \cite{Li2017}. In \ac{MAC} with correlated sources, as shown in Figure \ref{fig:intro mac corr scr diagram}, there are multiple transmitters, each observing a source, and the sources are correlated with each other. The transmitters wish to send their observations in a distributed fashion via a MAC to a central receiver. The receiver reconstructs the sources losslessly. \ac{CES} showed that joint source-channel coding outperforms separation-based coding approaches \cite{Ahlswede1973,Slepian-Wolf_MAC}.  This was done by introducing a novel transmission scheme \cite{CES}, which exploits the common information between the sources. 
In this scheme,  \ac{GKW} common part between the sources 
is first extracted distributively at the encoders. 
The encoders can effectively `fully cooperate' to send this information to the receiver,  
as it is done in \ac{ptp} joint source-channel coding problem.
The rest of the sources are transmitted using distributed unstructured  random mappings. In summary, it employs a two-stage encoding strategy. 
\ac{CES} also characterized a set of sufficient conditions, in terms of single-letter information quantities, for transmission of sources over a \ac{MAC}. The scheme is known to be suboptimal \cite{Dueck} in general. There are a set of necessary conditions developed in \cite{Lapidoth_MAC_scr} and \cite{Lapidoth_dependence_balance}. However, characterizing the optimal necessary and sufficient conditions for transmission of discrete memoryless sources over \ac{MAC} is still an open problem. 

\begin{figure}[hbtp]
\centering
\includegraphics[scale=1]{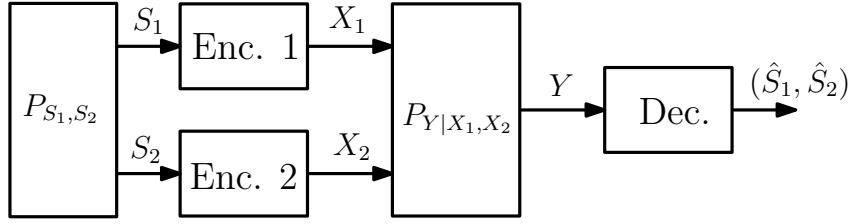}
\caption{A schematic of a two-user MAC with correlated sources. In this setup, the source sequences $(S_1^n, S_2^n)$ are observed by the corresponding encoders. The encoders produce $(X_1^n,X_2^n)$ which are channel input sequences. Upon observing the channel output $Y^n$, the decoder produces an estimate for the sources. 
}
\label{fig:intro mac corr scr diagram}
\end{figure}

Another fundamental problem in which common information plays a key role is communication of \emph{independent} messages over discrete memoryless \ac{MAC-FB}. In a \ac{MAC-FB} setup (see Figure \ref{fig:MAC-FB two-user}), after each channel use, the output of the channel is received at each transmitter noiselessly. This problem has been studied extensively in the literature \cite{Gaarder-Wolf,Cover-Leung,Kramer-thesis,Ramji-Sandeep-FB,Willems-FB,Willems_partialFB,Schalkwijk-Kailath,Ozarow}. Gaarder and Wolf \cite{Gaarder-Wolf} showed that feedback can expand the capacity region of discrete memoryless MAC. 
\ac{CL} \cite{Cover-Leung} studied two-user \ac{MAC-FB}, developed a coding strategy using unstructured random codes, and characterized an achievable rate region in terms of single-letter information quantities. Later, it was shown by Willems \cite{Willems-FB} that the CL scheme achieves the feedback capacity for a class of \ac{MAC-FB}. However, this is not the case for general MAC-FB \cite{Ozarow}. There are several improvements over CL achievable region, namely \cite{Lapidoth} and \cite{Ramji-Sandeep-FB}. A multi-letter characterization of the feedback-capacity of MAC-FB is given by Kramer \cite{Kramer-thesis}. However, the characterization is not computable, since it is an infinite-letter characterization. Finding a computable characterization of the capacity region remains an open problem.

\begin{figure}
    \centering
    \includegraphics[scale=0.9]{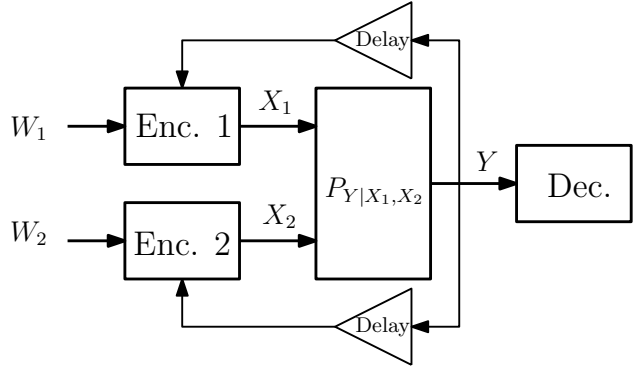}
    \caption{A schematic of a two-user MAC with feedback setup. The output of the channel is available, with one unit of delay, to the transmitters.}
    \label{fig:MAC-FB two-user}
\end{figure}

The main idea behind CL coding scheme is explained in the following. The scheme operates in two stages. In stage one, the transmitters send the messages with rates that lie outside the no-feedback capacity region (i.e. higher rates than what is achievable without feedback).  The transmission rates are taken such that each user can decode the other user's message using feedback. In this stage, the receiver is unable to decode the messages reliably; however, is able to form a list of ``highly likely" pairs of messages. The transmitters can also recreate this list. In the second stage, the encoders fully cooperate to send the index of the correct message-pair in the list, and help the receiver decode it. 

There is a connection between CES scheme for transmission of correlated sources over MAC and CL scheme for communications over MAC-FB. In a MAC-FB setup, after multiple uses of the channel, conditioned on feedback, the messages become statistically correlated. As explained above, at the end of the first stage in CL scheme, the messages are decoded at the transmitters. Hence, the decoded messages can be viewed as a GKW common part available at the two transmitters after the first stage. This common part is used in the second stage to resolve the uncertainty of the receiver. In connection with CES scheme, the common part is transmitted using identical random unstructured codebooks. 

In this work, we study three-user MAC with correlated sources, and  three-user MAC-FB with independent messages. Motivated by the notion of common information and its imperative role in these problems,  we start by identifying common information among a triplet of sources (say $S_1, S_2, S_3$). One can extend \ac{GKW} common part to define a (mutual) common part for $(S_1,S_2,S_3)$ in a straightforward way. In addition, one can define the pairwise GKW common parts between any pair $(S_i,S_j)$ as a part of the common information. The mutual common part together with the pairwise common parts characterize a vector of four components of common information which we refer to as \textit{univariate common parts}. 

We make the following contributions in this work. We, first, identify a new additional structure in the joint probability distribution of the sources, called ``conferencing common part''. This common part can be viewed as the GKW common part between a source (say $S_1$) and a pair of sources (say $S_2, S_3$). More explicitly, it is defined as the random variable $T$ with the largest entropy for which there exist a function $f(\cdot)$ and a bivariate function $g(\cdot, \cdot)$ such that $T=f(S_1)=g(S_2,S_3)$ with probability one.  Therefore, for the triplet $(S_1,S_2,S_3)$, there are three conferencing common parts, \add{one between each source and the other pair}. We also refer to these as bivariate common parts. Hence, in total, we identify the common parts among a triplet of the sources as a vector of seven components, including four univariate and three conferencing (bivariate) common parts.

Next, we develop a new coding strategy to exploit a particular form of the conferencing common parts among the sources, one given by additive functions. Efficient encoding of conferencing common parts is a more challenging task as compared to the univariate ones --- which is done using identical random unstructured mappings/codebooks.  This is because conferencing common parts are not available at any one transmitter--- rather a conference among a subset of the users is needed to extract these common parts. We develop a multiuser joint-source channel coding methodology  based on structured mappings to encode these common parts efficiently to be transmitted over a MAC.  

In particular, we design coding strategies based on random structured mappings for three-user MAC with correlated sources and MAC-FB. 
For the former problem, our  coding strategy exploits the univariate and the conferencing common information among the sources. We derive a new set of sufficient conditions for this coding strategy using single-letter information quantities for arbitrary sources and channel distributions.   For the latter problem, 
based on our notion for common information, we develop a new coding scheme for communications over three-user MAC-FB with independent messages. We characterize its performance using single-letter information quantities and derive an inner bound to the capacity region. 
For both problems we provide a set of examples, 
where these rate regions are shown to be optimal. Moreover, 
we analytically prove that this performance is not achievable using random unstructured mappings/codes. The main results of this paper are given in Proposition \ref{pro:linear cods better than CES} and Theorem \ref{them: achievable-rate-for -proposed scheme}-\ref{thm:MAC-FB structured}.

\noindent \textbf{Prior works on structured codes for multiuser problems:} Structured codes have been used in many problems involving either source coding or channel coding. For example, they have been used in  distributed source coding \cite{Dinesh_dist_source_coding,Ahlswede-Han,Han_Kobayashi_DSC,Han1987}, computation over MAC \cite{Nazer_Gasper_Comp_MAC,Arun_comp_over_MAC_ISIT13,Zhan2013,Appuswamy2013}, MAC with side information  \cite{Philosof2011,Philosof-Zamir,Arun_MAC_states_IT,Ahlswede-Han,ISIT17_MAC_States}, interference channels  \cite{Vishwanath_Jafar_Shamai_2008,hong_caire,Bresler2010a,Niesen2013,Jafarian2012,Ordentlich2012},  and broadcast channels  \cite{Arun_BC_18}.

\noindent\textbf{Notations:}
In this paper, random variables are denoted using capital letters such as $X,Y$, and their realizations are shown using lower case letters such as $x,y$, respectively.  Vectors are shown using lowercase bold letters such as $\mathbf{x}, \mathbf{y}$.  Calligraphic letters are used to denote sets such as $\mathcal{X}, \mathcal{Y}$. For any set $\mathcal{A}$, let $S_{\mathcal{A}}=\{S_a\}_{a\in \mathcal{A}}$. If $\mathcal{A}=\emptyset$, then $S_{\mathcal{A}}=\emptyset$.
  As a shorthand, we sometimes denote a triple $(s_1,s_2,s_3)$ by $\underline{s}$. We also denote a triple of sequences $(\mathbf{s}_1,\mathbf{s}_2,\mathbf{s}_3)$ by $\underline{\mathbf{s}}$. \add{Binary entropy function is denoted by $h_b(\cdot)$.}
By $\FF_q$, we denote the field of integers modulo-$q$, where $q$ is a prime number.  \add{Modulo-$q$ addition is denoted by $\oplus_q$, and, when it is clear from the context, the subscript $q$ is removed.}
For any mapping $\Phi: \mathcal{A} \mapsto \mathcal{B}$ and any integer $n$, define the mapping $\Phi^n: \mathcal{A}^n \mapsto \mathcal{B}^n$ such that $\Phi^n(a^n)\deq (\Phi(a_1), \Phi(a_2), ..., \Phi(a_n))$ for all $a^n\in \mathcal{A}^n$. 
Given a probability distribution $P_X$ on a finite alphabet $\mathcal{X}$, let $A_{\epsilon}^{(n)}(X)$ denote the set of strongly $\epsilon$-typical sequences of length $n$. We follow the definition of typical sequences as given in \cite{Csiszar,ElGamal-book}.

The rest of the paper is organized as follows: Section \ref{sec:Prelim} contains problem formulation and known results for MAC with correlated sources. We present our contributions for this problem in  Section \ref{sec:MAC_scr structured}. Similarly, we present the problem formulation and known results for MAC-FB in Section \ref{sec:MAC-FB prel}, and provide our contributions for this problem in \ref{sec:MAC-FB structured}. Lastly, Section \ref{sec: conclusion} concludes the paper.

\section{Transmission of Sources Over MAC: Preliminaries}\label{sec:Prelim}


\subsection{Problem Formulation}

As depicted in Figure \ref{fig:intro mac corr scr diagram}, the problem of MAC with correlated sources consists of multiple transmitters, each observing a source sequence statistically correlated to others. The source sequences are sent by the encoders  via a MAC to a central decoder. The objective of the receiver is to reconstruct the source sequences losslessly. It is assumed that the channel is a discrete memoryless MAC and the source sequences are discrete and  generated IID according to a known joint PMF. In what follows, we formulate this problem more precisely. 
%


\begin{definition}\label{def:DMS MAC}
A discrete memoryless MAC with $3$ users is defined by input alphabet $\mathcal{X}_1 \times \mathcal{X}_2 \times \mathcal{X}_3$, output alphabet $\mathcal{Y}$, and a transition probability matrix $P_{Y|X_1,X_2, X_3}$. The input and output alphabets are assumed to be finite sets.  The MAC is denoted by the triple $(\underline{\mathcal{X}}, \mathcal{Y}, P_{Y|\underline{X}}).$
\end{definition}
We assume that the channel is memoryless, stationary and used without feedback, and, hence, the transition probability of the $n$-length channel output vector given the $n$-length channel input vectors is given by 
\begin{equation*}
\prod_{i=1}^n P_{Y|X_1X_2X_3}(y_i|x_{1i}, x_{2i}, x_{3i}),
\end{equation*}
for all $\underline{\bfx}\in \underline{\mathcal{X}}^n$ and $\bfy\in \mathcal{Y}^n$.

%
%

\begin{definition} 
A discrete memoreless stationary source $(S_1,S_2,S_3)$ is defined by alphabet $\mathcal{S}_1 \times \mathcal{S}_2\times \mathcal{S}_3$ and a distribution $P_{S_1, S_2,S_3}$. The source is denoted by the pair $(\underline{S},  P_{\underline{S}})$
\end{definition}
The distribution of $n$-length source sequences is given by 
\begin{equation*}
\prod_{i=1}^n P_{S_1S_2S_3}(s_{1i}, s_{2i},s_{3i}),
\end{equation*}
for all $\underline{\mathbf{s}}\in \underline{\mathcal{S}}^n$.

In this paper, the bandwidth expansion factor is assumed to be unity, i.e., the channel is used $n$ times for transmission of $n$ samples of the sources. 
\begin{definition}
A coding scheme (without bandwidth expansion) with parameter $n$ for transmission of a source $(\underline{S},  P_{\underline{S}})$ over a MAC $(\underline{\mathcal{X}}, \mathcal{Y}, P_{Y|\underline{X}})$ consists of encoding functions $e_i:\mathcal{S}_i^n \rightarrow \mathcal{X}_i^n, i=1,2,3$, and a decoding function $d:\mathcal{Y}^n \rightarrow \mathcal{S}^n_1 \times \mathcal{S}^n_2 \times \mathcal{S}^n_3$. The parameter $n$ is called blocklength. 
\end{definition}

\begin{definition}
A source  $(\underline{S},  P_{\underline{S}})$ is said to be transmissible over a MAC $(\underline{\mathcal{X}}, \mathcal{Y}, P_{Y|\underline{X}})$, if for all $\epsilon>0$ and for all sufficiently large $n$, there exists a coding scheme with parameter $n$ such that
\[
\sum_{\underline{\mathbf{s}} \in \underline{\mathcal{S}}^n} P_{\underline{S}}^n(\underline{\mathbf{s}}) 
\sum_{\mathbf{y}:d(\mathbf{y}) \neq \underline{\mathbf{s}}}
P_{Y|\underline{X}}^n \Big(\mathbf{y}~|~ \mathbf{x}_i=e_i(\mathbf{s}_i),~ i=1,2,3\Big)\leq \epsilon.
\]
\end{definition}

%
%
%
%
%
%

\subsection{\ac{CES} Sufficient Conditions: Two-User Case}
The two-user version of MAC with correlated sources was investigated in \cite{CES} and \ac{CES} scheme was proposed based on unstructured random mappings. Further, a sufficient condition for transmissibility is derived in terms of single-letter information quantities. 
In this scheme the notion of \ac{GKW} \textit{common part} plays an important role. The formal definition of such common part and the \ac{CES} sufficient conditions are given below.

\begin{definition}[GKW Common part]\label{def: comm part}
A \textit{common part} between random variables $(S_1,S_2)$ is a random variable $W_{12}$ with the largest entropy for which there exist functions $f,g$ such that $W_{12}=f(S_1)$, and $W_{12}=g(S_2)$ with probability one. In this work, such a random variable $W_{12}$ is called a univariate common part.  
\end{definition}



\begin{fact}[\ac{CES} sufficient conditions]
A source $(\mathcal{S}_1, \mathcal{S}_2, P_{S_1S_2})$ is transmissible over a MAC $(\mathcal{X}_1, \mathcal{X}_2, \mathcal{Y}, P_{Y|X_1X_2} )$, if there exist distributions $P_{U_{12}}, P_{X_1|S_1,U_{12}}$ and $P_{X_2|S_2,U_{12}}$  such that,
\begin{align*}
H(S_1|S_2) & \leq I(X_1;Y|X_2, S_2,   U_{12}),\\
H(S_2|S_1) & \leq I(X_2;Y|X_1, S_1, U_{12}),\\
H(S_1,S_2|W_{12}) & \leq I(X_1X_2;Y|W_{12}, U_{12}),\\
H(S_1, S_2)  & \leq I(X_1 X_2;Y),
\end{align*} 
where, $U_{12}$ is an auxiliary random variable with a finite alphabet $\mathcal{U}_{12}$, and the joint distribution of all the random variables factors as  $$P_{S_1,S_2,U_{12},X_1,X_2, Y}=P_{S_1,S_2}P_{U_{12}} P_{X_1|S_1,U_{12}}P_{X_2|S_2,U_{12}} P_{Y|X_1,X_2}.$$
\end{fact}


\subsection{A Sufficient Condition Based on Unstructured Mappings: Three-User Case}
One can extend \ac{CES} sufficient conditions for three-user case based on unstructured random codes. For that, first we need to generalize the definition of GKW common part for more than two random variables. 

\begin{definition} 
The common part among random variables $(S_1, S_2, S_3)$ is the random variable $W_{123}$ with the largest entropy for which there exist functions $f_i, i=1,2,3$ such that $W_{123}=f_i(S_i)$ holds with probability one. 
\end{definition}

It is worth noting that for the triple $(S_1,S_2,S_3)$ there are four common parts namely $(W_{12}, W_{13}, W_{23}, W_{123})$. For the case of multiple sources, say $(S_1,S_2,S_3)$, a similar idea as in CES can be used to encode the univariate common parts. In what follows we provide an extension of CES scheme to three-use case based on unstructured random mappings. 

\begin{definition} 
Given a source  $(\underline{\mathcal{S}}, P_{\underline{S}})$ and a MAC $(\underline{\mathcal{X}}, \mathcal{Y}, P_{Y|X_1X_2X_3} )$, let $\mathscr{P}_{CES}$ be the set of conditional distributions $P_{\underline{U}, \underline{X}| \underline{S}}$ defined on $\underline{\mathcal{U}}\times   \underline{\mathcal{X}}$ which factors as
\begin{align}\label{eq: CES scr joint dist}
P_{U_{123}}\Bigg[\prod_{b\in \{12,13,23\}}P_{U_b|W_{b}U_{123}}\Bigg]~ \Bigg[  \prod_{\substack{i, j, k \in \{1,2,3\}\\j<k, i\neq j, i\neq k}}   P_{X_i|S_iU_{123}U_{ij}U_{ik}}\Bigg],
\end{align}
where, with a slight abuse of notation,  $\underline{U}\deq (U_{123}, U_{12}, U_{13}, U_{23})$ and its alphabet is a finite set denoted by $\underline{\mathcal{U}}$.
\end{definition}

\begin{proposition}\label{prep: CES_three_user}
A source $(\underline{\mathcal{S}}, P_{S_1S_2S_3})$ is transmissible over a $\MAC$, if there exists a conditional distribution $P_{\underline{U}, \underline{X}| \underline{S}}\in \mathscr{P}_{CES}$ such that for any distinct $i,j,k\in \{1,2,3\}$ and any $\mathcal{B}\subseteq \{12 ,13,23\}$ the following inequalities hold
\begin{align*}
H(S_i|S_j S_k )&\leq I(X_i;Y|S_jS_kX_jX_k  U_{123} U_{12}U_{13} U_{23}),\\
H(S_iS_j|S_k  )&\leq I(X_iX_j;Y|S_k   U_{123}U_{ik}U_{jk}X_k),\\
H(S_iS_j|S_k W_{ij})&\leq I(X_iX_j;Y|S_k W_{ij} U_{123} U_{12}U_{13} U_{23}X_k),\\
H(S_1S_2S_3|W_{123} W_{\mathcal{B}})&\leq I(X_1X_2X_3;Y|W_{123}W_{\mathcal{B}}U_{123}U_{\mathcal{B}}),\\
H(S_1S_2S_3)&\leq I(X_1X_2X_3;Y),
\end{align*}
where we have identified $U_{ij}=U_{ji}$ and $W_{ij}=W_{ji}$. 
\end{proposition}


The three-user extension of CES involves three layers of coding. In the first layer $W_{123}$ is encoded at each transmitter to $U_{123}$. Next, based on the output of the first layer, $W_{ij}$'s are encoded to $U_{ij}$.  Finally, based on the output of the first and the second layers, $S_1,S_2$ and $S_3$ are encoded. Figure \ref{fig: CES 3 RVs} shows the random variables involved in the extension of CES.

\begin{figure}[hbtp]
\centering
\includegraphics[scale=0.7]{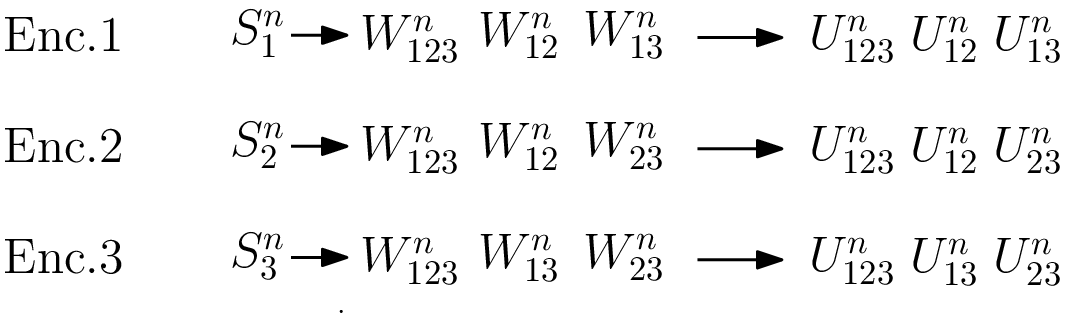}
\caption{The random variables involved in the three-user extension of CES.}
\label{fig: CES 3 RVs}
\end{figure}


\begin{proof}[Outline of the proof]

Fix a conditional distribution $P_{\underline{U}, \underline{X}| \underline{S}}\in \mathscr{P}_{CES}$. Let the sequence $\mathbf{s}_i\in \mathcal{S}_i^n$ be a realization of the $i$th source, where $i=1,2,3$.

\noindent {\textbf{Codebook Generation:}}
The construction of the codebooks at each transmitter is given below:
\begin{enumerate}
\item For each realization $\mathbf{w}_{123}$ of the mutual common part, a sequence $\bfU_{123}$ is generated randomly according to $\prod_{l\in [1,n]} P_{U_{123}}$. Such a sequence is indexed by $\bfU_{123}(\mathbf{w}_{123})$. 

\item Given  $b \in\{12,13,23\}$, and for each $\mathbf{u}_{123}$ and $\mathbf{w}_{b}$, a sequence $\bfU_b$ is generated randomly according to $\prod_{l\in [1,n]} P_{U_b|W_b U_{123}}$. Such a sequence is indexed by $\bfU_{b}(\mathbf{w}_{b}, \mathbf{u}_{123})$.

\item Given distinct elements  $i,j,k \in \{1,2,3\}$, any realization $\mathbf{s}_i$ of the source, the common parts  $(\mathbf{w}_{123}, \mathbf{w}_{ij}, \mathbf{w}_{ik})$, and the corresponding sequences $\bfU_{123}(\mathbf{w}_{123}),\allowbreak \bfU_{ij}(\mathbf{w}_{ij}, \bfU_{123})$ and $~\bfU_{ik}(\mathbf{w}_{ik}, \bfU_{123})$, a sequence $\bfX_i$ is generated randomly according to  $\prod_{l\in [1,n]} P_{X_i|S_i U_{123} U_{ij} U_{ik} }$. For shorthand, such  a sequence is denoted by $\bfX_i(\mathbf{s}_i, \bfU_{123},\allowbreak \bfU_{ij},\bfU_{ik})$.

\end{enumerate}

\noindent {\textbf{Encoding:}}
Upon observing a realization $\mathbf{s}_i$ of the $i$th source, transmitter $i$ first calculates the common part sequences $(\mathbf{w}_{123}, \mathbf{w}_{ij}, \mathbf{w}_{ik})$, where $i,j,k \in \{1,2,3\}$ are distinct. Then, the transmitter finds the corresponding sequences $$(\bfU_{123}(\mathbf{w}_{123}), \bfU_{ij}(\mathbf{w}_{ij}, \bfU_{123}), \bfU_{ik}(\mathbf{w}_{ik}, \bfU_{123}))$$ and sends $\bfX_i(\mathbf{s}_i, \bfU_{123},\bfU_{ij},\bfU_{ik})$ over the channel.

\noindent {\textbf{Decoding:}}
Upon receiving the channel output sequence $\mathbf{y}$, the decoder finds a unique triple $(\tilde{\mathbf{s}}_1,\tilde{\mathbf{s}}_2,\tilde{\mathbf{s}}_3)$ such that 
\begin{align*}
(\underline{\tilde{\mathbf{s}}}, \tilde{\bfU}_{123}, \tilde{\bfU}_{12}, \tilde{\bfU}_{13}, \tilde{\bfU}_{23}, \tilde{\bfX}_1,\tilde{\bfX}_2,\tilde{\bfX}_3,\mathbf{y})\in A_{\epsilon}^{(n)}(\underline{S},U_{123},U_{12},U_{13},U_{23},X_1,X_2,X_3,Y),
\end{align*}
where $ \tilde{\bfU}_{123}=\bfu_{123}(\tilde{\mathbf{w}}_{123}),\tilde{\bfU}_{ij}=\bfu_{ij}(\tilde{\mathbf{w}}_{ij}, \tilde{\bfU}_{123})$, $\tilde{\bfX}_i=\bfX_i(\tilde{\mathbf{s}}_i, \tilde{\bfU}_{123},\tilde{\bfU}_{ij},\tilde{\bfU}_{ik})$, and $i,j,k \in \{1,2,3\}$ are distinct. Note that $(\tilde{\mathbf{w}}_{123}, \tilde{\mathbf{w}}_{12},\tilde{\mathbf{w}}_{13},\tilde{\mathbf{w}}_{23})$ are the corresponding common parts sequences of $(\tilde{\mathbf{s}}_1,\tilde{\mathbf{s}}_2,\tilde{\mathbf{s}}_3)$.

A decoding error will be occurred, if no unique $(\tilde{\mathbf{s}}_1,\tilde{\mathbf{s}}_2,\tilde{\mathbf{s}}_3)$ is found. Using a standard argument as in \cite{CES}, it can be shown that the  probability of error can be made sufficiently small for large enough $n$, if the conditions in Proposition \ref{prep: CES_three_user} are satisfied.

\end{proof}

\section{Transmission of Sources Over MAC: Structured Mappings }\label{sec:MAC_scr structured}

In this section, we provide a new sufficient condition characterized using single-letter information quantities for transmissibility of the sources over MAC using structured mappings. The main results of this section are given in Proposition \ref{pro:linear cods better than CES}, Theorem \ref{them: achievable-rate-for -proposed scheme} and \ref{thm:CES is subopt}.  

\subsection{Conferencing Common Information}\label{sec:conf comm info}
The joint distribution of triple $(S_1,S_2,S_3)$ also has an additional structure which is not captured by the univariate common parts defined previously. This will be addressed by defining a new common part as follows.

\begin{definition} \label{def: bivariat comm}
The \textit{conferencing} common part of a triple of random variables $(S_1,S_2,\allowbreak S_3)$ is the triple of random variables $(T_1,T_2,T_3)$ with the largest joint entropy, for which there exist functions $f_i, g_i, i\in \{1,2,3\}$ such that $T_i = f_i(X_i)=g_i(X_j, X_k)$ hold with probability one for all distinct $i,j,k \in \{1,2,3\}$\footnote{{ Note that the conferencing common part random variables are unique upto a relabeling.}}.
\end{definition}
From definitions \ref{def: comm part} and \ref{def: bivariat comm}, the common parts among the three random variables $(S_1, S_2,S_3)$ are $(W_{12}, W_{13}, W_{23}, W_{123}, T_1,T_2,T_3)$, where $W_{ij}$ is the pairwise common part between $(S_i,S_j)$, $W_{123}$ is the mutual common part (all in the sense of Definition \ref{def: comm part} ), and $(T_1,T_2,T_3)$ are conferencing common parts (as in Definition \ref{def: bivariat comm}) among $(S_1,S_2,S_3)$.
In this work, we focus on a special class of conferencing common part which is defined as follows.

\begin{definition}\label{def: m-additive}
The additive common part of a triple of random variables $(S_1,S_2,S_3)$ is the triple of random variables $(T_1,T_2,T_3)$ with the largest entropy for which there exist a finite field $\FF_q$ and functions $f_i, i=1,2,3$ such that $T_i\in \FF_q$, $ T_1\oplus_q T_2 \oplus_q T_3=0$.
%
%
\end{definition}
%

The following example provides a triplet of binary sources with additive common part where the associated finite field is $\FF_2$.

\begin{example}\label{ex: bivariate}
Let $S_1,S_2$ and $S_3$ be three Bernoulli random variables. Suppose $S_1$ and $S_2$ are independent, with biases $p_1$ and $p_2$, respectively, and $S_3=S_1 \oplus_2 S_2$ with probability one. It is not difficult to show that univariate common parts are trivial, i.e., $(W_{12},W_{13},W_{23}, W_{123})$ is a constant.  As for the conferencing common parts, set $T_i=S_i, i=1,2,3$. Then $(T_1,T_2,T_3)$ satisfies the conditions in Definition \ref{def: m-additive} for $q=2$. Therefore,  $(T_1,T_2,T_3)$  is the additive common part of $(S_1,S_2,S_3)$. 
 \end{example}


Unlike univariate common information, conferencing common parts are not available at any terminal. This is due to the fact that conferencing common parts are bivariate functions of the sources. As a result, to exploit conferencing common information, a new coding technique needs to be developed. For this purpose, we use affine maps. The key concepts are described in the following.

We construct three affine maps for encoding of such common parts. Let $\mathbf{G}$ be a $n$ by $n$ matrix with elements in $\FF_q$. We, also, select vectors $\mathbf{b}_1, \mathbf{b}_2, \mathbf{b}_3\in \FF_q^n$ such that $\mathbf{b}_1\+\mathbf{b}_2 \oplus \mathbf{b}_3=\mathbf{0}$. The additive common parts are encoded as  $\bfV_i^n=\bfT_i^n\mathbf{G}\+\mathbf{b}_i$, for  $i=1,2,3$, and hence, the equality $\bfV^n_1\oplus \bfV^n_2 \oplus \bfV^n_3 =\mathbf{0}$ holds with probability one. One may adopt a randomized affine map to encode the additive common parts. For that, we can select the matrix $\mathbf{G}$ and the vectors $\mathbf{b}_1, \mathbf{b}_2, \mathbf{b}_3$ randomly and uniformly from the set of all matrices and vectors with elements in $\FF_q$.

\subsection{Sub-optimality of Unstructured Mappings}\label{subsec:subopt unstructures}
     In what follows, we show that applications of affine maps for transmission of additive common parts improves upon the scheme based on unstructured random mappings given in the previous section. 

\begin{example}\label{ex: additive source MAC}
Suppose $(S_1, S_2, S_3)$ are as in Example \ref{ex: bivariate}. The sources are to be transmitted via a MAC with binary inputs $\mathcal{X}_1\times \mathcal{X}_2 \times \mathcal{X}_3 $, binary outputs $\mathcal{Y}_1 \times \mathcal{Y}_2$, and a conditional probability distribution that satisfies
\begin{align}
 (Y_1,Y_2)=\begin{cases} 
      (X_1\+N_\delta, X_2\+N'_{\delta}), &\text{if}~ X_3 = X_1\+ X_2, \\
      (N_{1/2}, N'_{1/2}), &  \text{if}~ X_3 \neq X_1\+ X_2,
   \end{cases}
\end{align}
where $N_{\delta}, N'_{\delta}, N_{1/2}$ and $N'_{1/2}$ are independent Bernoulli random variables with parameter $\delta, \delta, \frac{1}{2}$, and $\frac{1}{2}$, respectively. 

As explained in Example \ref{ex: bivariate}, the univariate common parts are trivial, and the $2$-additive common parts are $T_i=S_i, i=1,2,3$. For such a setup, we use random affine maps explained above. The following lemma provides a necessary and sufficient condition for reliable transmission of $(S_1,S_2,S_3)$. The achievability is obtained using the above approach. 
\end{example} 
\begin{proposition}\label{pro:linear cods better than CES}
Consider the source given in Example \ref{ex: bivariate} with $p_1=p_2=p$. Such a source is transmissible over the MAC given in Example \ref{ex: additive source MAC}, if and only if $h_b(p)\leq 1-h_b(\delta), i=1,2 $. Moreover, the source with parameter $p=h_b^{-1}(1-h_b(\delta))$ does not satisfy the sufficient condition in Proposition \ref{prep: CES_three_user}. 
\end{proposition}

\begin{proof}
The proof for the direct part follows using random affine maps. { For that, set $X_i^n=S_i^n\mathbf{G}\+\mathbf{B}_i, i=1,2,3$, where $\mathbf{G}, \mathbf{B}_1, \mathbf{B}_2, \mathbf{B}_3$ are selected randomly, and uniformly with elements from $\FF_q$ and satisfying $\mathbf{B}_1\+\mathbf{B}_2\+ \mathbf{B}_3=\mathbf{0}$.  In this case, $X_3^n=X_1^n\+X_2^n$ which implies that  $Y_1^n=X_1^n\+N_\delta^n$ and $Y_2^n=X_2^n\+N_{\delta'}^n$. Hence, from the properties of random linear maps for the point-to-point joint source-channel setting, $(S_1,S_2)$ can be decoded with arbitrary small error probability, if $h_b(p_i)\leq 1-h_b(\delta), i=1,2$. }

For the converse part, suppose $(S_1,S_2,S_3)$ are transmissible. Therefore, for any $\epsilon>0$ there exists a coding scheme with error probability at most $\epsilon$. Suppose  $(e_1,e_2,e_3)$ are the encoders  and $d$ is the decoder of such a scheme. Then, from Fano's inequality,
\begin{align*}
2h_b(p)=\frac{1}{n}H(S_1^n,S_2^n)&\leq \frac{1}{n} I(S_1^n,S_2^n;Y_1^n,Y_2^n) +2\epsilon +\frac{1}{n}h_b(\epsilon)\\
&\stackrel{(a)}{\leq}\frac{1}{n}I( X_1^n, X_2^n,X_3^n;Y_1^n,Y_2^n) +2\epsilon +\frac{1}{n}h_b(\epsilon)\\
&\stackrel{(b)}{\leq} 2-2h_b(\delta)+2\epsilon +\frac{1}{n}h_b(\epsilon),
\end{align*}
where $(a)$ follows because of the Markov chain $(S_1,S_2,S_3)\leftrightarrow (X_1,X_2,X_3)\leftrightarrow(Y_1,Y_2)$. Inequality $(b)$ holds as the mutual information does not exceed the sum-capacity of the MAC which equals to $2-2h_b(\delta)$. The proof for the converse is complete as the inequalities hold for arbitrary $\epsilon>0$. 

Next, we prove the last statement of the proposition by contradiction. Suppose the sources  with parameter $p_1=p_2=h_b^{-1}(1-h_b(\delta))$ satisfy the conditions in Proposition \ref{prep: CES_three_user}. Then, from the fourth inequality in Proposition \ref{prep: CES_three_user},
\begin{equation*}
2-2h_b(\delta)\leq \max_{P_{\underline{U}, \underline{X}| \underline{S}}\in \mathscr{P}_{CES}}I(X_1,X_2,X_3;Y|\underline{U})= \max_{P_{\underline{U}} P_{\underline{X}|\underline{U}\underline{S}}} I(X_1X_2 X_3;Y|\underline{U}).
\end{equation*}
where $P_{\underline{X}|\underline{U}\underline{S}}=\prod_{i=1}^3 P_{X_i|S_i, \underline{U}}$. The equality holds as there is no univariate common part, and hence, $\underline{U}$ is independent of the sources. Since,  $\underline{U}$ appears in the conditioning in the mutual information term, the above inequality is equivalent to 
\begin{equation}
2-2h_b(\delta) \leq \max_{P_{X_1|S_1}P_{X_2|S_2}P_{X_3|S_3}} I(X_1X_2 X_3;Y).
\end{equation}
One can verify that $I(X_1,X_2,X_3;Y)\leq 2-2h_b(\delta)$, with equality, if and only if, $X_3=X_1\+X_2$ with probability one, and $X_1$ and $X_2$ are uniform over $\{0,1\}$. However, we show that such distribution cannot be generated by taking the marginal of $P_{\underline{S}}P_{X_1|S_1}P_{X_2|S_2}P_{X_3|S_3}$. This is because, to get $X_1$ and $X_2$ to be uniform over $\{0,1\}$, we need to set $P_{X_1|S_1}(x|s)=P_{X_2|S_2}(x|s)=\frac{1}{2}$ for all $x,s \in \{0,1\}$. This implies that, $X_1$ and $X_2$ are independent of each other and of $S_1$ and $S_2$, respectively. Hence, $P_{\underline{S}, \underline{X}}=P_{\underline{S}}P_{X_1}P_{X_2}P_{X_3|S_3}$, which means that $(X_1,X_2)$ are independent of $X_3$. This contradicts with the condition that $X_3=X_1\+X_2$.
\end{proof}

\subsection{New Sufficient Condition}
We use the intuition behind the argument in Subsection \ref{subsec:subopt unstructures} and propose a new coding strategy in which a combination of random linear codes (as in Example \ref{ex: additive source MAC}) and the extension of CES scheme is used. The coding scheme uses both univariate and additive common information among the sources. In the next Theorem, we derive sufficient conditions for transmission of correlated sources over three-user MAC.

\begin{definition} 
Given a source  $\sourceS$ with an additive common part $(T_1,T_2,T_3)$, and a MAC $\MAC$, let $\mathscr{P}$ be the set of conditional distributions $P_{\underline{U}, \underline{V}, \underline{X}| \underline{S}}$ defined on $\underline{\mathcal{U}}\times \FF_q^3 \times  \underline{\mathcal{X}}$ which can be factored as
\begin{align}\label{eq: thm mac corr scr joint dist}
P_{U_{123}}\Bigg[\prod_{b\in \{12,13,23\}}P_{U_b|W_{b}U_{123}}\Bigg]~P_{V_1V_2V_3}~ \Bigg[  \prod_{\substack{i, j, k \in \{1,2,3\}\\j<k, i\neq j, i\neq k}}   P_{X_i|S_iU_{123}U_{ij}U_{ik}V_i}\Bigg],
\end{align}
where $\FF_q$ is the finite field associated with the additive common part, the random variables $(W_{123}, W_{12}, W_{13}, W_{23})$ are the univariate common parts of the sources,  $P_{V_1V_2V_3}=\frac{1}{q^2}\11\{V_3\+_qV_1\oplus_q V_2=0\}$, and with slight abuse of notation $\underline{U}\deq (U_{123}, U_{12}, U_{13}, U_{23})$. $\underline{\mathcal{U}}$ and $\underline{\mathcal{V}}$ are finite alphabets associated with the auxiliary random variables $\underline{U}$ and $\underline{V}$, respectively. 
\end{definition}

\begin{theorem}\label{them: achievable-rate-for -proposed scheme}
A source  $\sourceS$ with an additive common part $(T_1,T_2,T_3)$  is reliably transmissible over a MAC $\MAC$, if there exists a conditional distribution $P_{\underline{U}, \underline{V}, \underline{X}| \underline{S}}\in \mathscr{P}$ such that for all $a,b\in \FF_q$, any distinct $i,j,k \in \{1,2,3\}$, and for any $\mathcal{B}\subseteq \{12,13, 23\}$ the following inequalities hold:
{ \begin{subequations}\label{eq: transmittable bounds}
\begin{align}
H(S_i|S_j,S_k)&\leq I(X_i;Y|S_j, S_kU_{123}, U_{12}, U_{13}, U_{23}, V_1, V_2, V_3, X_j, X_k)\\
H(S_i,S_j|S_k, W_{\mathcal{B}})&\leq I(X_i,X_j;Y|S_k, W_{\mathcal{B}}, U_{123}, U_{ik}, U_{jk}U_{\mathcal{B}},  V_{k}, X_k)\\
H(S_i,S_j|S_k, W_{\mathcal{B}}, \underline{T})&\leq I(X_i,X_j;Y|S_k, W_{\mathcal{B}}, U_{123}, U_{ik}, U_{jk}U_{\mathcal{B}}, \underline{T}, \underline{V}, X_k)\\
H(S_1, S_2, S_3|W_{123}, W_\mathcal{B}, \underline{T} )&\leq I(X_1, X_2, X_3; Y|W_{123}, W_\mathcal{B}, U_{123}, U_\mathcal{B}, \underline{T}, \underline{V})\\
H(S_1, S_2, S_3|\underline{T})&\leq I(X_1, X_2, X_3; Y|\underline{T}, \underline{V})\\
H(S_1, S_2, S_3|a T_1\+_q bT_2)&\leq I(X_1, X_2, X_3;Y|a T_1\+_q bT_2, a V_1\+_qbV_2)\\
H(S_1, S_2, S_3|W_{123},W_\mathcal{B}, aT_1\+_qbT_2 )&\leq I(X_1, X_2, X_3;Y|W_{123}, W_\mathcal{B}, U_{123}, U_\mathcal{B}, aT_1\+_qbT_2, aV_1\+_qbV_2)
\end{align}
\end{subequations} }
\end{theorem}


\begin{remark}\label{rem:mac corr scr}
The set of sufficient conditions given in Theorem \ref{them: achievable-rate-for -proposed scheme} includes the one in Proposition \ref{prep: CES_three_user}. For that select the joint distribution in \eqref{eq: thm mac corr scr joint dist} such that $X_i$ be independent of $V_i$ for all $i=1,2,3$. 
\end{remark}
%
%
%

\begin{proof}[Outline of the proof]
We use a new approach which is based on affine maps to encode additive common parts. Suppose the random variables $(\underline{S},\underline{X},U_{123}, U_{12},U_{13}, U_{23}, \underline{V})$ are distributed according to a joint distribution that factors as in \eqref{eq: thm mac corr scr joint dist}.

\noindent {\textbf{Codebook Generation:}}
At each transmitter five different codebooks are defined, one codebook for the additive common part $T_i$, three codebooks for univariate common parts $(W_{123}, W_{ij}, W_{ik})$, where $i,j, k$ are distinct elements of $\{1,2,3\}$, and one codebook for generating the total output $X_i^n$.
Fix $\epsilon>0$. 
\begin{enumerate}

\item The codebooks for encoding of univariate common parts are as in the proof of Proposition \ref{prep: CES_three_user}. 

\item The codebook for encoding of $(T_1,T_2,T_3)$ is defined using affine maps. Generate two vectors $\mathbf{B}_1, \mathbf{B}_2$ of length $n$, and an $n \times n$ matrix $\mathbf{G}$ with elements selected randomly, uniformly and independently from $\FF_q$. Set $\mathbf{B}_3=-(\mathbf{B}_1\oplus_q \mathbf{B}_2)$.  For each sequence $\mathbf{t}_i\in \FF_q^n$, define $\bfV_i(\mathbf{t}_i)=\mathbf{t}_i \mathbf{G}\oplus \mathbf{B}_i$, where $i=1,2,3$, and all the additions and multiplications are modulo-$q$.

\item Given distinct $i,j,k \in \{1,2,3\}$, any realization $\mathbf{s}_i$ of the source, the common parts  $(\mathbf{w}_{123}, \mathbf{w}_{ij}, \mathbf{w}_{ik}, \mathbf{t}_i)$, and the corresponding sequences $$\big(\bfU_{123}(\mathbf{w}_{123}), \bfU_{ij}(\mathbf{w}_{ij}, \bfU_{123}),\bfU_{ik}(\mathbf{w}_{ik}, \bfU_{123}), \bfV_i(\mathbf{t}_i)\big)$$ generate a random IID sequence $\bfX_i$ according to  $\prod_{l\in [1,n]} P_{X_i|S_i U_{123} U_{ij} U_{ik} V_i }$. For shorthand, such a sequence is denoted by $\bfX_i(\mathbf{s}_i, \bfU_{123},\bfU_{ij},\bfU_{ik}, \mathbf{V}_i)$.
\end{enumerate}

\noindent {\textbf{Encoding:}}
Assume $\mathbf{s}_i$ is a realization of the $i$th source, where $i=1,2,3$. Transmitter $i$ first calculates the common part sequences $(\mathbf{w}_{123}, \mathbf{w}_{ij}, \mathbf{w}_{ik}, \mathbf{t}_i)$, where $i,j,k \in \{1,2,3\}$ are distinct. Next, the transmitter finds the corresponding sequences $$\big(\bfU_{123}(\mathbf{w}_{123}), \bfU_{ij}(\mathbf{w}_{ij}, \bfU_{123}),\bfU_{ik}(\mathbf{w}_{ik}, \bfU_{123}), \bfV_i(\mathbf{t}_i)\big)$$ and sends $\bfX_i(\mathbf{s}_i, \bfU_{123},\bfU_{ij},\bfU_{ik}, \mathbf{V}_i)$ to the channel. 

\noindent {\textbf{Decoding:}}
Upon receiving the channel output vector $\mathbf{y}$ from the channel, the decoder finds sequences $\tilde{\mathbf{s}}_i \in \mathcal{S}_i^n, i=1,2,3$, such that 
\begin{align}\label{eq:MAC scr decoder}
(\underline{\tilde{\mathbf{s}}}, \tilde{\bfU}_{123}, \tilde{\bfU}_{12}, \tilde{\bfU}_{13}, \tilde{\bfU}_{23}, \underline{\tilde{\mathbf{v}}}, \underline{\tilde{\bfX}},\mathbf{y})\in A_{\epsilon}^{(n)}(\underline{S},U_{123},U_{12},U_{13},U_{23},\underline{V}, \underline{X},Y),
\end{align}
where $ \tilde{\bfU}_{123}=\bfu_{123}(\tilde{\mathbf{w}}_{123}),\tilde{\bfU}_{ij}=\bfu_{ij}(\tilde{\mathbf{w}}_{ij}, \tilde{\bfU}_{123}),  \tilde{\mathbf{v}}_i=\mathbf{v}_i(\tilde{\mathbf{t}}_i)$, $\tilde{\bfX}_i=\bfX_i(\tilde{\mathbf{s}}_i, \tilde{\bfU}_{123},\tilde{\bfU}_{ij},\tilde{\bfU}_{ik}, \tilde{\mathbf{t}}_i)$, and $i,j,k \in \{1,2,3\}$ are distinct. Note that $(\tilde{\mathbf{w}}_{123}, \tilde{\mathbf{w}}_{12},\tilde{\mathbf{w}}_{13},\tilde{\mathbf{w}}_{23})$ and $(\tilde{\mathbf{t}}_1, \tilde{\mathbf{t}}_2, \tilde{\mathbf{t}}_3)$ are the univariate and additive common part sequences of $(\tilde{\mathbf{s}}_1,\tilde{\mathbf{s}}_2,\tilde{\mathbf{s}}_3)$, respectively.

A decoding error will be occurred, if no unique $(\tilde{\mathbf{s}}_1,\tilde{\mathbf{s}}_2,\tilde{\mathbf{s}}_3)$ is found. It is shown in Appendix \ref{sec: proof of them lin_ces} that the probability of error approaches zero as $n\rightarrow \infty$, if the inequalities in (\ref{eq: transmittable bounds}) are satisfied. 
\end{proof}

\begin{remark}
The coding strategy explained in the proof of Theorem \ref{them: achievable-rate-for -proposed scheme} subsumes the extension of CES scheme and identical random linear coding strategy.   
\end{remark}


\subsection{Example with Structural Mismatch} \label{sec: improvements over CES}

%

{\ In Example \ref{ex: additive source MAC}, the structure in the sources matches with that of the channel. In other words, the source correlation is 
captured via the relation given by $S_3=S_1\+S_2$, and when 
$X_3=X_1 \+ X_2$, the channel behaved obligingly. 
In this section, we consider an example where there is a mismatch between the structures of the source and the channel. In other words, the source correlation is still governed by $S_3=S_1\+S_2$, whereas, the channel fuses $X_3$ and $X_1 \+ X_2$ 
in a nonlinear fashion. In what follows, we provide an application of our coding scheme in scenarios where there is a structural mismatch between the sources and the channel. 
%
}
%
%

\begin{example}\label{ex: CES is suboptimal}
Consider the sources denoted by $(S_1, S_2, S_3)$, where  $S_1$ and $S_3$ are  independent Bernoulli random variables with parameter $\sigma, \gamma\in [0,\frac{1}{2}]$, respectively. Suppose the third source satisfies  $S_3=S_1\oplus_2 S_2$ with probability one. For shorthand we associate such sources with the parameters $(\sigma, \gamma)$.    The sources are to be transmitted trough a MAC with binary inputs as shown in Figure \ref{fig: Exp1}. In this channel the noise random variable $N$ is assumed to be independent of other random variables. The PMF of $N$ is given in Table \ref{tab: N}, {where the parameter $ \delta\in (0, \frac{1}{4}]$. As a result, $H(N)=1+\frac{1}{2}h_b(2\delta)$. }
\begin{table}
\caption {Distribution of $N$}\label{tab: N}
\begin{center}
\begin{tabular}{c|c|c|c|c}
N & 0 & 1 & 2 & 3\\
\hline
$P_N$ & $\frac{1}{2}-\delta$ & $\frac{1}{2}$ & $\delta$ & $0$
\end{tabular}
\end{center}
\end{table}

\begin{figure}
\centering
\includegraphics[scale=0.8]{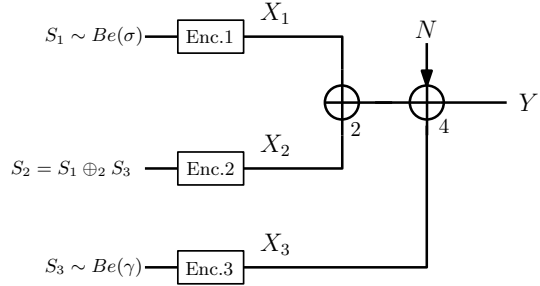}
\caption{The diagram the setup introduced in Example \ref{ex: CES is suboptimal}. Note the input alphabets of this MAC are restricted to $\{0,1\}$.}
\label{fig: Exp1}
\end{figure}
\end{example}

For this setup,  we show that there exist parameters $(\sigma, \gamma)$ whose corresponding sources in Example \ref{ex: CES is suboptimal} cannot be transmitted reliably using the CES scheme. However,  according to Theorem \ref{them: achievable-rate-for -proposed scheme}, such sources can be reliably transmitted. This emphasizes the fact that efficient encoding of  conferencing common information  contributes to improvements upon coding schemes solely based on univariate common information. In what follows, we explain the steps to show the existence of such parameters.

\begin{remark}\label{rem: gamma}
For the special case in which $\sigma=0$, the equalities $S_1=0$ and $S_2=S_3$ hold  with probability one. From Proposition \ref{prep: CES_three_user}, such $(S_1,S_2,S_3)$ can be transmitted using CES scheme, if $h_b(\gamma)\leq 2-H(N)$ holds. 
\end{remark}

Let $\gamma^*\in [h_b^{-1}(0.5),\frac{1}{2})$ be such that $\gamma^*=h_b^{-1}(2-H(N))$. Such a  $\gamma^*$ exists as $2-H(N)=1-\frac{1}{2}h_b(2\delta)$ and, thus, is a number between $\frac{1}{2}$ to $1$. By Remark \ref{rem: gamma}, the sources $(S_1,S_2,S_3)$ with parameter $(\sigma=0,\gamma=\gamma^*)$ can be transmitted reliably using CES scheme. However, we argue that for small enough $\epsilon>0$, the sources with parameter $(\sigma=\epsilon, \gamma=\gamma^*-\epsilon)$ cannot be transmitted using this scheme. Whereas, from Theorem \ref{them: achievable-rate-for -proposed scheme}, this source can be transmitted reliably. This is formally stated as follows.

\begin{theorem}\label{thm:CES is subopt}
There exist $\sigma \in (0, \frac{1}{2}]$ and $\gamma \in (0, \gamma^*]$ such that the triplet sources $(S_1,S_2,S_3)$ with these parameters satisfies the sufficient condition of Theorem \ref{them: achievable-rate-for -proposed scheme}, thus, transmissible over the channel in Example \ref{ex: CES is suboptimal}, but does not satisfy the sufficient condition in Proposition \ref{prep: CES_three_user}. 
\end{theorem}

\begin{proof}
The proof is in Appendix \ref{appx: proof_ thm_CES is sub opt}. 
\end{proof}

\section{Communications over MAC with Feedback: Preliminaries}\label{sec:MAC-FB prel}
The problem of three user MAC with noiseless feedback is depicted in Figure \ref{fig: MAC with FB}. This communication channel consists of one receiver and multiple transmitters. After each channel use, the output of the channel is received at each transmitter noiselessly. Gaarder and Wolf \cite{Gaarder-Wolf} showed that the capacity region of the MAC can be expanded through the use of the feedback. This was shown in a binary erasure MAC.
Cover and Leung \cite{Cover-Leung} studied the two-user MAC with feedback, and developed a coding strategy using unstructured random codes.

\begin{figure}[hbtp]
\centering
\includegraphics[scale=0.9]{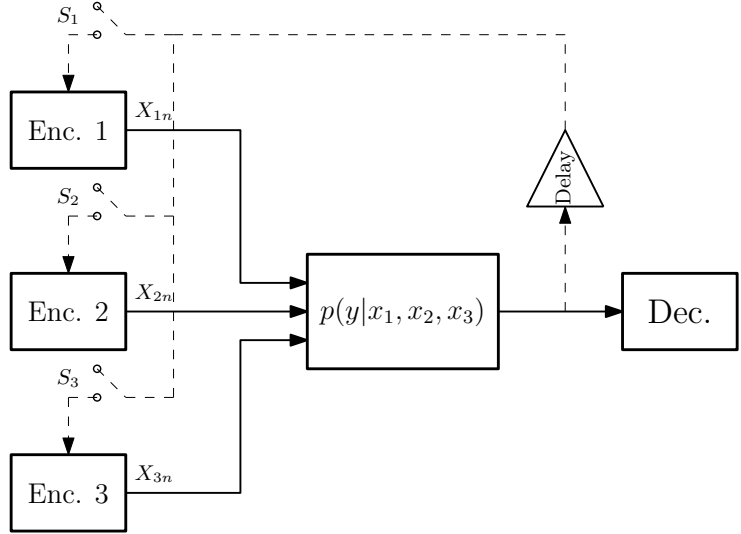}
\caption{The three-user MAC with noiseless feedback. If the switch $S_i$ is closed, the feedback is available at the $i$th encoder, where $i=1,2,3$.}
\label{fig: MAC with FB}
\end{figure}

\subsection{Model and Problem Formulation}\label{sec: prelim}
In what follows, we formulate the problem of communications over MAC-FB. We restrict ourselves to three-user MAC with noiseless feedback in which all or a subset of the transmitters have access to the feedback perfectly. Consider a three-user MAC identified by a transition probability matrix  $P_{Y|X_1, X_2, X_3}$ as in Definition \ref{def:DMS MAC}. Let $\bfy^n$ be a realization of the output of the channel after $n$ uses, where $\bfx^n_i$ is the $i$th input sequence of the channel, $i\in [1,3]$. Then, the conditional probability distribution of the channel output $y_n$ given the current and past input and output vectors is given by
{\revcol \begin{align}\label{eq: chann probabilities}
P_{Y_n|\bfY^{n-1}, \bfX^n_1, \bfX^n_2, \bfX^n_3}(y_n|\bfy^{n-1}, \bfx_1^n, \bfx^n_2, \bfx^n_3)=P_{Y|X_1, X_2,  X_3}(y_n|x_{1n}, x_{2n}, x_{3n}).
\end{align} }
It is assumed that noiseless feedback is made available, with one unit of delay, to a subset $\mathcal{T} \subseteq [1,3]$ of the transmitters. In Figure \ref{fig: MAC with FB}, the switches $S_i, i=1,2,3$ determine which transmitter receives the feedback. A formal definition of a MAC-FB setup is given in the following. 

\begin{definition} \label{def:MAC_FB setup}
A $3$-user MAC-FB setup is characterized by a $3$-user MAC $\MAC$ and a subset $\mathcal{T} \subseteq [1,3]$ determining the transmitters which have access to the feedback. It is assumed that at least one transmitter has access to the feedback, i.e., $|\mathcal{T}|\geq 1$. {\revcol Such a MAC-FB is denoted by $(\underline{\mathcal{X}}, \mathcal{Y}, P_{Y|\underline{X}}, \mathcal{T})$.}
\end{definition}

\begin{definition}\label{def:MAC_FB scheme}
For a $3$-user MAC-FB $(\underline{\mathcal{X}}, \mathcal{Y}, P_{Y|\underline{X}}, \mathcal{T})$, an $(N, \Theta_1, \Theta_2, \Theta_3)$ coding scheme consists of $3$ sequences of encoding functions defined as,  
\begin{align*}
e_{i,n}: [1, \Theta_i] \times \mathcal{Y}^{n-1}\rightarrow \mathcal{X}_i, \quad \text{for}~ i\in \mathcal{T}, \quad \text{and} \quad 
e_{j,n}: [1, \Theta_j] \rightarrow \mathcal{X}_j, \quad \text{for}~ j\in \mathcal{T}^c,  
\end{align*}
where $n\in [1, N]$ and a decoding function denoted by 
\begin{align*}
d: \mathcal{Y}^N\rightarrow [1,\Theta_1] \times [1, \Theta_2] \times [1, \Theta_3].
\end{align*}
We use a unified notation $e_{i,n}(m, y^{n-1})$ to denote the encoders, as it is understood  that for $i \notin \mathcal{T}$ the encoder $e_{i,n}$ is only a function of the message $m$. Moreover, for shorthand, the encoders of the coding scheme are denoted by $\underline{e}$.
\end{definition}
It is assumed that, transmitter $i$ receives a message index $M_i$ which is drawn randomly and uniformly from $[1,\Theta_i]$, where $i\in [1,3]$.  Furthermore, the message indexes $(M_1, M_2, M_3)$ are assumed to be mutually independent. 
For this setup, the average probability of error is defined as 
\begin{align}\label{eq:MAC_FB P_e}
P_{err}(\underline{e}) \deq   \PP\{d(Y^N) \neq (M_1, M_2, M_3)\}, 
\end{align} 
where $\underline{e}$ denotes the encoders of the coding scheme. 
\begin{definition}\label{def: achievable rate MAC_FB}
For a $3$-user MAC-FB, a rate-tuple $(R_1,R_2, R_3)$ is said to be achievable, if for any $\epsilon >0$ there exists, for all sufficiently large $N$, an $(N, \Theta_1, \Theta_2, \Theta_3)$ coding scheme such that 
\begin{align*}
P_{err}(\underline{e})  <\epsilon, \quad \frac{1}{N}\log_2 \Theta_i \geq R_i-\epsilon, \quad \text{where}~~i\in [1,3].
\end{align*} 
\end{definition}

\subsection{\ac{CL} Achievable Region: Unstructured Coding Approach }
  The main idea behind the CL scheme is to use superposition block-Markov encoding. The scheme operates in two stages. In stage one, the transmitters send the messages with rates outside the no-feedback capacity region, but small enough that 
each user can decode the other user's message using feedback.   In the second stage, the encoders fully cooperate to send the messages to disambiguate the information at the receiver. 
Using this approach, the following rate-region is achievable  for communications over a MAC with noiseless feedback available at at least on of the transmitters \cite{Cover-Leung}.
\begin{fact}
Given a two-user MAC-FB $(\mathcal{X}_1, \mathcal{X}_2, \mathcal{Y}, P_{Y|X_1,X_2}, \mathcal{T}\subseteq \{1,2\})$, a rate pair $(R_1,R_2)$ is achievable, if there exist distributions $P_U, P_{X_1|U}$, and  $P_{X_2|U}$ such that 
\begin{align*}
R_1 \leq I(X_1; Y|X_2, U), \qquad R_2 \leq I(X_2; Y|X_1, U),\qquad R_1+R_2 \leq I(X_1,X_2; Y), 
\end{align*} 
where $U$ takes values from a finite set $\mathcal{U}$, and the joint distribution of all the random variables factors as $P_UP_{X_1|U}P_{X_2|U}P_{Y|X_1,X_2}$.
\end{fact}

It was shown in \cite{Willems_partialFB} that, in a two-user MAC-FB, the \ac{CL} rate region is achievable even if only one of the transmitters has access to the feedback ( $|\mathcal{T}|=1$).

As explained in \ac{CL} scheme, the decoded sub-messages $(M_{1,b}, M_{2,b})$ are used as a common information for the next block of transmission. One can extend this scheme for a multi-user MAC-FB setup (say a three-user MAC-FB) using unstructured codes.  In this setup, the transmitters send the messages with rates outside the no-feedback capacity region. Hence, the receiver is not able to decode the messages. However, the transmission rates are taken to be sufficiently low so that each user can decode the sub-messages of the other users. The decoded sub-messages at the end of each block $b$ are used as uni-variate common parts for the next block of transmission. One can derive a single-letter characterization of an achievable rate region based on such a scheme in a straightforward fashion. For conciseness we do not state this rate region in this paper. 


\section{Three-User MAC-FB: Structured Codes}\label{sec:MAC-FB structured}

In this section, we propose a new coding scheme for  three-user MAC-FB, and derive a computable single-letter achievable rate region  (an inner bound to the capacity region) using structured codes -- in particular, \textit{quasi-linear} codes that were introduced in \cite{QLC-ISIT16}. Note that prior to the start of the communication, the messages are mutually independent; whereas after multiple uses of the channel, they become statistically correlated conditioned on the feedback. Based on this observation, we make a connection to the problem of MAC with correlated sources to design coding strategies that exploit the statistical correlation among the messages. We use the notion of conferencing common information to propose a new coding strategy for $3$-user MAC-FB. The main results of this section are given in Theorem \ref{thm: MAC-FB achievable} and \ref{thm:MAC-FB structured}.

\subsection{New Achievable Rate Region}

In what follows, we give the intuition behind the use of conferencing common information in MAC-FB. Consider a three-user MAC-FB setup as depicted in Figure \ref{fig:MAC_FB conf comm}. Similar to the two-user version of the problem, the communications take place in $B$ blocks each of length $n$. Moreover, the message at Transmitter $i$ is divided into $B$ sub-messages denoted by $(M_{i,1}, M_{i,2}, ..., M_{i,B})$, where $i=1,2,3$. Suppose, the transmission rates are such that neither the decoder nor the transmitters can decode the messages.  However, at each block $b$, the rates are sufficiently low so that each transmitter is able to decode the modulo-$q$ sum of the other two sub-messages\footnote{{\revcol It is understood that the messages belong to a finite field $\FF_q$.}}. For instance, Transmitter 1 can decode $M_{2,b}\oplus M_{3,b}$ with high probability. Let $T_{i,b}$ denote the decoded sum at Transmitter $i$, where $i=1,2,3$. Then, for binary messages, $T_{1,b}\oplus T_{2,b} \oplus T_{3,b}=0$ with high probability. As a result, $(T_1,T_2,T_3)$ can be interpreted as additive conferencing common parts (see Definition \ref{def: m-additive}). Building upon this intuition, in what follows,  we propose a coding strategy for communications over 3-user MAC-FB. Further, we derive a new commutable achievable rate region for the three-user MAC with feedback problem.

\begin{figure}[hbtp]

\centering
\includegraphics[scale=1]{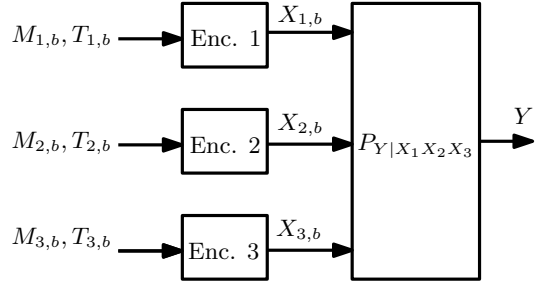}
\caption{Applications of conferencing common information for communications over MAC-FB. The new sub-messages at block $b$ are denoted by $M_{i,b}$. At the end of block $b-1$, each transmitter decode the modulo-two sum of the other two transmitters. The decoded sums are denoted by $T_{i,b}, i=1,2,3$. Note that $T_{1,b}\oplus T_{2,b} \oplus T_{3,b}=0$ with probability close to one.}
\label{fig:MAC_FB conf comm}
\end{figure}


{\revcol
We start by the following definition to characterize an achievable rate region.}
 \begin{definition}\label{def:MAC-FB P}
For a prime $q$ and a given set  $\mathcal{U}$ and a three-user MAC-FB $\MACFB$, define $\mathscr{P}$ as the collection of all distributions on $\mathcal{U}\times \FF_q^6 \times \underline{\mathcal{X}}, {\mathcal{Y}}$ factoring as
\begin{align}\label{eq:MAC-FB P}
P_U P_{V_1V_2V_3}\prod_{i=1}^3 P_{T_i} P_{X_i|UT_i V_i} P_{Y|X_1X_2X_3},
\end{align}
where  $(T_1, T_2, T_3)$ are mutually independent with uniform distribution over a finite field $\FF_q$, $(V_1, V_2, V_3)$ are pairwise independent each with uniform distribution over $\FF_q$, and $P_{V_1V_2V_3}(v_1,v_2,v_3)=\frac{1}{q^2}\11\{ v_1\oplus v_2 \oplus v_3=0\}$, {\revcol and for any $i\in \mathcal{T}^c$, we have $P_{X_i|UT_i V_i}=P_{X_i}$ for some distribution on $\mathcal{X}_i$.}
\end{definition}

{
Fix a distribution $\P\in \mathscr{P} $ that factors as in \eqref{eq:MAC-FB P}.
Denote $S_i=(X_i, T_i, V_i)$ for $i=1,2,3$. Consider two sets of random variables $(U, S_1,S_2,S_3, Y)$ and $(\tilde{U},\tilde{S}_1, \tilde{S}_2, \tilde{S}_3, \tilde{Y})$. {\revcol We describe the joint distribution of these random variables.} The distribution of each set of the random variables is $\P$, i.e.,  
\begin{align*}
P_{U S_1 S_2 S_3  Y}=P_{\tilde{U} \tilde{S}_1 \tilde{S}_2\tilde{S}_3 \tilde{Y}}=\P.
\end{align*}
In addition, conditioned on $(\tilde{U},\tilde{S}_1, \tilde{S}_2, \tilde{S}_3, \tilde{Y})$ we have 
\begin{align}\label{eq:P P tilde joint dist}
P_{U S_1 S_2 S_3 Y| \tilde{U} \tilde{S}_1 \tilde{S}_2\tilde{S}_3 \tilde{Y}}=P_U P_{V_1V_2V_3| \tilde{T}_1 \tilde{T}_2\tilde{T}_3 }\prod_{i=1}^3 P_{T_i} P_{X_i|UT_i V_i} P_{Y|X_1X_2X_3},
\end{align}
with $\underline{V}=\underline{\tilde{T}}\bf A$ with probability one, {\revcol where $\bf A$ is a $3\times 3$ matrix with elements in $\FF_q$ and }the multiplications are modulo $q$. Further, $\bf A$ is chosen such that $P_{V_1V_2V_3}=P_{\tilde{V}_1,\tilde{V}_2, \tilde{V}_3}$. 

\begin{definition} \label{def:MAC-FB rate}
Given a MAC-FB $\MACFB$, let $\mathcal{R}_{\text{MAC-FB}}$ be the set of triplets $(R_1,R_2,R_3)$ for which there exist $\alpha\in (0,1)$, random variables $(U, S_1,S_2,S_3, Y)$ and $(\tilde{U},\tilde{S}_1, \tilde{S}_2, \tilde{S}_3, \tilde{Y})$ distributed according to \eqref{eq:P P tilde joint dist} for some $P\in\mathscr{P} $ and matrix $\mathbf{A}\in \FF_q^{3\times 3}$ and mutually independent random variables $(W_1,W_2,W_3)$ which are also independent of other random variables such that the following inequalities hold for any subset $\mathcal{B}\subseteq \{1,2,3\}$ and any  distinct elements $i, j, k \in \{1,2,3\}$:  
\begin{align*}
\alpha H(W_i)&= R_i,\\
\alpha H(W_{\mathbf{A}_i}|W_i) &\leq I(T_{\mathbf{A}_i}; Y|U T_i V_i X_i),\\
\alpha H(W_j,W_k|W_{\mathbf{A}_i}, W_i)& \leq I(\tilde{T}_j \tilde{X}_j \tilde{T}_k \tilde{X}_k; Y \tilde{Y}| \tilde{U} \tilde{S}_i {U} {S}_i  \tilde{V}_j \tilde{V}_k ),\\
\alpha H(W_{\mathcal{B}}) &\leq I(X_\mathcal{B}; Y| U S_{\mathcal{B}^c} \tilde{V}_1, \tilde{V}_2, \tilde{V}_3)+ I(U;Y),
\end{align*}
{\revcol where $W_{\mathbf{A}_i}$ and $T_{\mathbf{A}_i}, i=1,2,3,$ are the $i$th element of the vector $\underline{W}\mathbf{A}$ and $\underline{T}\mathbf{A}$, respectively.}

\end{definition}

\begin{theorem}\label{thm: MAC-FB achievable}
For a MAC-FB $\MACFB$, the rate-region $\mathcal{R}_{\text{MAC-FB}}$  is achievable.
%
\end{theorem}
\begin{proof}
The proof is given in Appendix \ref{sec: thm 1}.
\end{proof}
}



\subsection{Necessity of Structured Codes for MAC-FB}

In this section, we show that coding strategies based on structured codes are necessary for certain instances of MAC with feedback. We first provide an example of a MAC with feedback. Then, we apply Theorem \ref{thm: MAC-FB achievable}  and  show that the inner bound achieves optimality.

\begin{example}\label{ex: example}
Consider the three-user MAC-FB problem depicted in Figure \ref{fig: Exp. setup}. In this setup, there is a MAC with three pairs of binary inputs, where the $i$th input is denoted by the pair $(X_{i1}, X_{i2})$ for $i=1,2,3$. The output of the channel is denoted by a binary vector $(Y_1, Y_{21}, Y_{22})$. Assume that noiseless feedback is available only at the third transmitter.  
\begin{figure}[hbtp]
\centering
\includegraphics[scale=0.9]{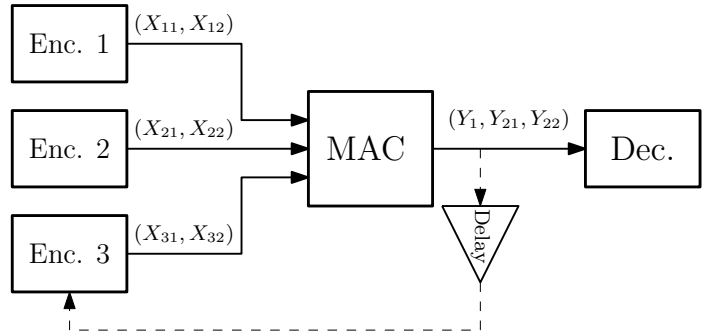}
\caption{The MAC with feedback setup for Example \ref{ex: example}.}
\label{fig: Exp. setup}
\end{figure}

The MAC in this setup consists of two parallel channels. The first channel is a three-user binary additive MAC with inputs $(X_{11}, X_{21}, X_{31})$, and output $Y_1$. The transition probability matrix of this channel is described by the following relation:
\begin{align*}
Y_1=X_{11}\oplus X_{21} \oplus X_{31}\oplus \tilde{N}_\delta,
\end{align*}
where $\tilde{N}_\delta$ is a Bernoulli random variable with bias $\delta$, and is independent of the inputs. The second channel is a MAC with $(X_{12}, X_{22}, X_{32})$ as the inputs, and  $(Y_{21},Y_{22})$ as the output. The conditional probability distribution of this channel satisfies
\begin{align}
 (Y_{21},Y_{22})=\begin{cases} 
      (X_{12}\+N_\delta, X_{22}\+N'_{\delta}), &\text{if}~ X_{32} = X_{12}\+ X_{22}, \\
      (N_{1/2}, N'_{1/2}), &  \text{if}~ X_{32} \neq X_{12}\+ X_{22},
   \end{cases}
\end{align}
where $N_{\delta}, N'_{\delta}, N_{1/2}$ and $N'_{1/2}$ are independent Bernoulli random variables with parameter $\delta, \delta, \frac{1}{2}$, and $\frac{1}{2}$, respectively.  The relation between the output and the input of the channel is depicted in Figure \ref{fig: Exp. second chann}. The channel operates in two states. If the condition $X_{31}=X_{12}\oplus X_{22}$ holds, the channel would be in the first state (the left channel in Figure \ref{fig: Exp. second chann}); otherwise it would be in the second state (the right channel in Figure \ref{fig: Exp. second chann}). In this channel, $N_\delta $ and $N'_\delta$ are Bernoulli random variables with identical bias $\delta$. Whereas, $N_{1/2} $ and $N'_{1/2}$ are Bernoulli random variables with bias $\frac{1}{2}$. We assume that $\tilde{N}_\delta, N_\delta, N'_\delta,N_{1/2} $, and $N'_{1/2}$ are mutually independent, and are independent of all the inputs. 
\begin{figure}[hbtp]
\centering
\includegraphics[scale=0.9]{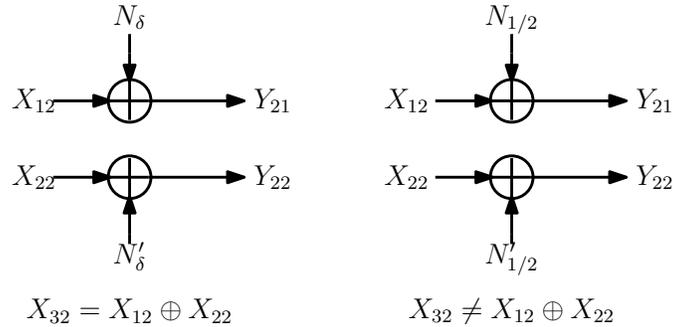}
\caption{The second channel for Example \ref{ex: example}. If the condition $X_{31}=X_{12}\oplus X_{22}$ holds, the channel would be the one on the left; otherwise it would be the right channel.}
\label{fig: Exp. second chann}
\end{figure}
\end{example}

We use linear codes to propose a new coding strategy for the setup given in Example \ref{ex: example}.  The scheme uses a large number $L$ of blocks{\revcol , each of length $n$}. Each encoder has two outputs, one for each channel. We use identical linear codes  with length $n$ and rate $\frac{k}{n}$ for each transmitter.  The coding scheme at each block is performed in two stages. In the first stage, each transmitter encodes the fresh message at the beginning of the block $l$, where $1\leq l \leq L$. The encoding process is performed using identical linear codes. At the end of block $l$, feedback is received by the third user. In stage 2, the third user uses the feedback from the first channel (that is $Y_1$) to decode the binary sum of the messages of the other encoders. Then, it encodes the summation, and sends it through its second output. If the decoding process is successful at the third user, then the relation $X_{32}=X_{12}\oplus X_{22}$ holds with probability one. This is because identical linear codes are used to encode the messages.  As a result of this equality, the channel in Figure \ref{fig: Exp. second chann} is in the first state with probability one. 
In the following theorem, we show that the rate $$(1-h(\delta), 1-h(\delta), 1-h(\delta))$$ is achievable using this strategy. Further, we prove in the followng theorem that any coding scheme achieving these rates must have codebooks that are almost closed under the binary addition. Since unstructured random codes do not have this property, any coding scheme solely based in them is suboptimal.  
%
%
%

\begin{theorem}\label{thm:MAC-FB structured}
For the channel given in Example \ref{ex: example}, the rate triple $(1-h(\delta), 1-h(\delta), 1-h(\delta))$ is achievable if and only if 1) user 3 decodes $X_1\+X_2$ with average probability of error approaching zero, and 2) the codebooks in user 1 and 2 must satisfy 
\begin{align*}
\lim_{N\rightarrow \infty}\frac{1}{N}\big| ~\log ||\mathcal{C}_{12}\oplus \mathcal{C}_{22}||- \log ||\mathcal{C}_{12}||~\big| =0, \quad \text{for} ~~i=1,2.
\end{align*}
\end{theorem}
\begin{proof}
The proof is given in Appendix \ref{appx:MAC-FB structured proof}.
\end{proof}

\section{Conclusion}\label{sec: conclusion}
A new form of common information, called ``conferencing common information'', is defined among triplets of random variables. Based on this notion, two coding strategies are proposed for three-user version of two problems: transmission of correlated sources over \ac{MAC}, and \ac{MAC} with feedback.  Further,  achievable rate regions of such strategies are characterized in terms of single-letter information quantities.  It is shown analytically that the proposed strategies outperform conventional unstructured random coding approaches in terms of achievable rates.    

~\\
\noindent \textbf{Acknowledgment:} We would like to thank Farhad Shirani of New York University for extensive and insightful discussions related to this work. 
 
\appendices

\section{Proof of Theorem \ref{them: achievable-rate-for -proposed scheme}}\label{sec: proof of them lin_ces}
\begin{proof}
There are two error events, $E_0$ and $E_1$. $E_0$ occurs if no  triple $\underline{\mathbf{\tilde{s}}}$ was found. $E_1$ occurs  if 
there exists $\underline{\mathbf{\tilde{s}}} \neq \underline{\mathbf{s}}$ such that equation \eqref{eq:MAC scr decoder} is satisfied. 
We consider a special case in which all the uni-variate common parts are trivial and that $T_i=S_i, i=1,2,3$. This implies that $S_1\+_qS_2\+_q S_3=0$ with probability one. The proof for the general case follows by adopting this proof and  the standard arguments as in \cite{CES}.

Suppose $\bf{v}_i(\cdot)$ and $\bf{x}_i(\cdot)$ are the realizations of random functions generated as in the outline of the proof of Theorem \ref{them: achievable-rate-for -proposed scheme}.  Using standard arguments one can show that $E_0\rightarrow 0$ as $n\rightarrow \infty$. We find the condition under which  $P(E_1 \cap E_0^c)\rightarrow 0$. For a given $\underline{\mathbf{s}}\in A_{\epsilon_1}(\underline{S})$, using the definition of $E_1$ and the union bound we obtain,
\begin{align*}
P(E_1 \cap E_0^c|\underline{\mathbf{s}}) \leq &  \sum_{ (\underline{\mathbf{v}},\underline{\mathbf{x}},\mathbf{y})\in  A_{\epsilon_2}(\underline{V},\underline{X},Y|\underline{\mathbf{s}})} \mathbbm{1}\{\mathbf{v}_i=\mathbf{v}_i(\mathbf{s}_i), \mathbf{x}_i=\mathbf{x}_i(\mathbf{s}_i, \mathbf{v}_i), ~ i=1,2,3\} P^n_{Y|\underline{X}}(\mathbf{y}|\underline{\mathbf{x}})\\
 &\sum_{\substack{(\underline{\mathbf{\tilde{s}}},\underline{\tilde{\mathbf{v}}},\underline{\tilde{\mathbf{x}}})\in A_{\epsilon_3}(\underline{S},\underline{V},\underline{X}|\mathbf{y})\\ \underline{\tilde{\mathbf{s}}} \neq \underline{\mathbf{s}}}} \mathbbm{1}\{\tilde{\mathbf{v}}_j=\bfv_j(\tilde{\mathbf{s}}_j), \tilde{\mathbf{x}}_j=\bfx_j(\tilde{\mathbf{s}}_j, \tilde{\mathbf{v}}_j), j=1,2,3\}
\end{align*}

Taking expectation over random vector functions  $\bfX_i(,)$ and $\bfV_i()$ gives,
\begin{align}\label{eq: Pe}
p_e(\underline{\mathbf{s}})=\EE\{P(E_{1}|\underline{\mathbf{s}})\} \leq &  \sum_{(\underline{\mathbf{v}}, \underline{\mathbf{x}},\mathbf{y})\in A_{\epsilon_2}(\underline{V}, \underline{X},Y|\underline{\mathbf{s}})} P^n_{Y|\underline{X}}(\mathbf{y}|\underline{\mathbf{x}})
 \sum_{\substack{(\underline{\mathbf{\tilde{s}}},\underline{\tilde{\mathbf{v}}},\underline{\tilde{\mathbf{x}}})\in A_{\epsilon_3}(\underline{S},\underline{V},\underline{X}|\mathbf{y})\\ \underline{\tilde{\mathbf{s}}} \neq \underline{\mathbf{s}}}} \\ \nonumber
& P\{\mathbf{v}_l=\bfV_l(\mathbf{s}_l), \mathbf{x}_l=\bfX_l(\mathbf{s}_l, \mathbf{v}_l), \tilde{\mathbf{v}}_l=\bfV_l(\tilde{\mathbf{s}}_l), \tilde{\mathbf{x}}_l=\bfX_l(\tilde{\mathbf{s}}_l, \tilde{\mathbf{v}}_l) ~ \mbox{for } l=1,2,3\}
\end{align}
Let 
\begin{equation}\label{eq: epsilon}
\epsilon=\max_{i\in [1,3]} \epsilon_i,
\end{equation}
where $\epsilon_i$ is as in the above summations.  
Note that $V_i(\cdot)$ and $\mathbf{X}_i(\cdot, \cdot)$ are generated independently. So the most inner term in (\ref{eq: Pe}) is simplified to 
\begin{align}\label{eq: pp 1}
P\{\mathbf{v}_j=\bfV_j(\mathbf{s}_j), \tilde{\mathbf{v}}_j=\bfV_j(\tilde{\mathbf{s}}_j)~ j=1,2\} P\{\mathbf{x}_l=\bfX_l(\mathbf{s}_l, \mathbf{v}_l), \tilde{\mathbf{x}}_l=\bfX_l(\tilde{\mathbf{s}}_l, \tilde{\mathbf{v}}_l) ~  l=1,2,3\}.
\end{align}
Note that $j=3$ is redundant because, $\mathbf{v}_3 \+_q \mathbf{v}_1\oplus_q \mathbf{v}_2=\bf0$ and $\mathbf{\tilde{v}}_3 \+_q \mathbf{\tilde{v}}_1\+_q \mathbf{\tilde{v}}_2=\bf0$. By definition, $\bfV_j(\mathbf{s}_j)=\mathbf{s}_j\mathbf{G}+\mathbf{B}_j, j=1,2$, where $\bfB_1,\bfB_2$ are uniform and  independent of $\mathbf{G}$. Then 
\begin{align}\label{eq: PP 2}
P\{\mathbf{v}_j=\bfV_j(\mathbf{s}_j), \tilde{\mathbf{v}}_j=\bfV_j(\tilde{\mathbf{s}}_j),~ j=1,2\} = \frac{1}{q^{2n}} P\{(\tilde{\mathbf{s}}_j-\mathbf{s}_j)\mathbf{G}=\tilde{\mathbf{v}}_j-\mathbf{v}_j, ~j=1,2\}
\end{align}
The following lemma determines the above term.
\begin{lem}\label{lem: prob of sG}
Suppose $\mathbf{G}$ is a $n\times m$ matrix with elements generated randomly and uniformly from $\FF_q$. If $\mathbf{s}_1$ or $\mathbf{s}_2$ is nonzero, the following holds: 
{ \begin{align}\label{eq: joint prob G}
P\{\mathbf{s}_j\mathbf{G}=\mathbf{v}_j, ~j=1,2\}= \left\{ \begin{array}{ll}
\mathbbm{1}\{\mathbf{v}_j=\mathbf{0}, ~ l=1,2\}, & \mbox{if}~ \mathbf{s}_1 = \mathbf{0},  \mathbf{s}_2=\mathbf{0}.\\
q^{-n}\mathbbm{1}\{\mathbf{v}_j=\mathbf{0}\} , & \mbox{if}~ \mathbf{s}_j=\mathbf{0}, \mathbf{s}_{j^c}\neq \mathbf{0}. \\
q^{-n}\mathbbm{1}\{\mathbf{v}_1 = a \mathbf{v}_2\}, & \mbox{if}~ \mathbf{s}_1 \neq \mathbf{0},  \mathbf{s}_2 \neq \mathbf{0}, \mathbf{s}_1= a \mathbf{s}_2,~ a\in \FF_q.\\
q^{-2n}, & \mbox{if}~ \mbox{otherwise}.\\
\end{array} \right.
\end{align} }
\end{lem}
\begin{proof}
We can write $\mathbf{s}_j\mathbf{G}=\sum_{i=1}^n \mathbf{s}_{ji}\mathbf{G}_i$, $j=1,2$, where $\mathbf{s}_{ji}$ is the $i$th component of $\mathbf{s}_j$ and $\mathbf{G}_i$ is the $i$th row of $\mathbf{G}$. Not that $\mathbf{G}_i$ are independent random variables with uniform distribution over $\FF_q^n$. Hence, if $\mathbf{s}_j\neq \mathbf{0}$, then $\mathbf{s}_j\mathbf{G}$ is uniform over $\FF_q^n$. Then, given the second condition in \eqref{eq: joint prob G}, $P\{\mathbf{s}_j\mathbf{G}=\mathbf{v}_j, ~j=1,2\}=q^{-n}\mathbbm{1}\{\mathbf{v}_j=\mathbf{0}\}$. {  ~If $\mathbf{s}_1=a\mathbf{s}_2$ with $a \in \FF_q$, then $\mathbf{s}_1\mathbf{G}=a \mathbf{s}_2\mathbf{G} $, with probability one and, thus, $P\{\mathbf{s}_j\mathbf{G}=\mathbf{v}_j, ~j=1,2\}=q^{-n}\mathbbm{1}\{\mathbf{v}_1 = a \mathbf{v}_2\}$.} 

{ If $\mathbf{s}_1\neq a \mathbf{s}_2$ for any $a\in \FF_q$, then $(\mathbf{s}_1, \mathbf{s}_2)$ are linearly independent. This implies that there exist indices $(l,k)$ such that the $2\times 2$ matrix $\mathbf{A}$ with elements $a_{11}=s_{1l}, a_{12}=s_{1,k}, a_{21}=s_{2l}$ and $a_{22}=s_{2k}$ is full rank.  As a result, $s_{1l} \mathbf{G}_l\+s_{1k}\mathbf{G}_k$ and $s_{2l} \mathbf{G}_l\+s_{2k}\mathbf{G}_k$ are independent random vectors with uniform distribution over $\FF^k_q$.} In this case, one can show that $\mathbf{s}_1\mathbf{G}$ is independent of $\mathbf{s}_2\mathbf{G}$. The proof follows by arguing that if a random variables $X$ is independent of $Y$ and is uniform over $\FF_q$, then $X\oplus_q Y$ is also uniform over $\FF_q$ and is independent of $Y$. 
\end{proof}
Finally, we are ready to characterize the conditions under which $p_e\rightarrow 0$. { Let $\mathcal{L}(\underline{\mathbf{s}})$ denote the set of all the variables $(\underline{\mathbf{v}}, \underline{\mathbf{x}},\mathbf{y}, \underline{\mathbf{\tilde{s}}},\underline{\tilde{\mathbf{v}}},\underline{\tilde{\mathbf{x}}})$ included in the summations in (\ref{eq: Pe}); more precisely, 
\begin{equation}\label{eq:L(s)}
\mathcal{L}(\underline{\mathbf{s}})\deq \Big\{ (\underline{\mathbf{v}}, \underline{\mathbf{x}},\mathbf{y}, \underline{\mathbf{\tilde{s}}},\underline{\tilde{\mathbf{v}}},\underline{\tilde{\mathbf{x}}}):  (\underline{\mathbf{v}}, \underline{\mathbf{x}},\mathbf{y})\in A_{\epsilon_2}(\underline{V}, \underline{X},Y|\underline{\mathbf{s}}), ~ (\underline{\mathbf{\tilde{s}}},\underline{\tilde{\mathbf{v}}},\underline{\tilde{\mathbf{x}}})\in A_{\epsilon_3}(\underline{S},\underline{V},\underline{X}|\mathbf{y}),
~\tilde{\mathbf{s}} \neq \underline{\mathbf{s}}\Big\}.
\end{equation}
 Based on the conditions in Lemma \ref{lem: prob of sG}, we partition this set into five subsets $\mathcal{L}_i(\underline{\mathbf{s}}), i=1,2,..., 5$. Hence, if $p_{e_i}(\underline{\mathbf{s}}), i\in [1,5]$ represents the contribution of each subset, then $p_e(\underline{\mathbf{s}})=\sum_{i=1}^5 p_{e_i}(\underline{\mathbf{s}})$. 
 In what follows,  we characterize these subsets and provide an upper bound to each term $p_{e_i}(\underline{\mathbf{s}}), i\in [1,5]$. }

\noindent{\bf Case 1, $\mathbf{\tilde{s}}_1 \neq \mathbf{s}_1,  \mathbf{\tilde{s}}_2= \mathbf{s}_2$:}

\hspace{10pt} In this case, using Lemma \ref{lem: prob of sG}, (\ref{eq: PP 2}) equals to $q^{-3n}\mathbbm{1}\{\mathbf{\tilde{v}}_2=\mathbf{v}_2\}$. As $\mathbf{s}_2=\mathbf{\tilde{s}}_2$ and $\mathbf{v}_2=\mathbf{\tilde{v}}_2$, then $X_2(\tilde{\mathbf{s}}_2, \tilde{\mathbf{v}}_2)=X_2(\mathbf{s}_2, \mathbf{v}_2)$. Therefore, we define 
\begin{align*}
\mathcal{L}_1(\underline{\mathbf{s}})\deq \Big\{ (\underline{\mathbf{v}}, \underline{\mathbf{x}},\mathbf{y}, \underline{\mathbf{\tilde{s}}},\underline{\tilde{\mathbf{v}}},\underline{\tilde{\mathbf{x}}})\in \mathcal{L}(\underline{\mathbf{s}}): \mathbf{\tilde{s}}_1 \neq \mathbf{s}_1,  \mathbf{\tilde{s}}_2= \mathbf{s}_2, \mathbf{\tilde{v}}_2= \mathbf{v}_2, \mathbf{x}_2=\mathbf{\tilde{x}}_2\Big\},
\end{align*}
where $\mathcal{L}(\underline{\mathbf{s}})$ is defined as in \eqref{eq:L(s)}.
Thus, the contribution of this case equals to 
\begin{align*}
p_{e_1} (\underline{\mathbf{s}}) \deq &\sum_{\mathcal{L}_1(\underline{\mathbf{s}})}  P^n_{Y|\underline{X}}(\mathbf{y}|\underline{\mathbf{x}})q^{-3n}  P\{\mathbf{x}_l=\mathbf{X}_l(\mathbf{s}_l, \mathbf{v}_l), \tilde{\mathbf{x}}_l=\mathbf{X}_l(\tilde{\mathbf{s}}_l, \tilde{\mathbf{v}}_l), ~  l=1,2,3\}.
\end{align*}
 

Note that $\mathbf{X}_l({\mathbf{s}}_l, \mathbf{{v}}_l)$ is independent of $\mathbf{X}_k(\tilde{\mathbf{s}}_k, \mathbf{\tilde{v}}_k)$, if $l \neq k$ or $\mathbf{s}_l \neq \mathbf{\tilde{s}}_l$ or $\mathbf{v}_l \neq \mathbf{\tilde{v}}_l$. Moreover, since $\mathbf{X}_l(\mathbf{s}_l, \mathbf{v}_l)$ is generated IID according to $P_{X_l|S_l,V_l}$, then for jointly typical sequences $(\mathbf{x}_l, \mathbf{s}_l,\mathbf{v}_l)$,
\begin{equation*}
\frac{-1}{n}\log_2 P\{\mathbf{x}_l=\mathbf{X}_l(\mathbf{s}_l, \mathbf{v}_l)\} \geq  H(X_l|S_lV_l))-\delta_1(\epsilon),
\end{equation*} 
where $\epsilon$ is defined as in \eqref{eq: epsilon} and $\delta_1(\epsilon)\geq 0$ is a continuous function satisfying $\lim_{\epsilon \rightarrow 0}\delta_1(\epsilon)= 0$.  Therefore, 
\begin{align*}
P\{\mathbf{x}_l=X_l(\mathbf{s}_l, \mathbf{v}_l), \tilde{\mathbf{x}}_l=X_l(\tilde{\mathbf{s}}_l, &\tilde{\mathbf{v}}_l) ~  l=1,2,3\}\\
&\leq 2^{-n[2H(X_1|S_1V_1)+H(X_2|S_2V_2)+ 2H(X_3|S_3V_3)-\delta_2(\epsilon)]} \mathbbm{1}\{\mathbf{\tilde{x}}_2=\mathbf{x}_2\},
\end{align*}
where $\delta_2$ is a non-negative and continuous function with $\lim_{\epsilon \rightarrow 0}\delta_2(\epsilon)= 0$. Note that for jointly typical sequences $(\mathbf{y}, \underline{\mathbf{x}})$, the conditional probability $P^n_{Y|\underline{X}}(\mathbf{y}|\underline{\mathbf{x}})$ is upper bounded by $2^{-n(H(Y|\underline{X})-\delta_3(\epsilon))}$.
Hence, we have:
\begin{align*}
p_{e_1}(\underline{\mathbf{s}})&\leq |\mathcal{L}_1(\underline{\mathbf{s}})| 2^{-nH(Y|\underline{X})}\frac{1}{q^{3n}} 2^{-n[2H(X_1|S_1V_1)+H(X_2|S_2V_2)+ 2H(X_3|S_3V_3)-\delta_4(\epsilon)]},
\end{align*}
where $\delta_4(\epsilon)\rightarrow 0$ as $\epsilon \rightarrow 0$ and $|\mathcal{L}_1(\underline{\mathbf{s}})|$ is the cardinality of $\mathcal{L}_1(\underline{\mathbf{s}})$. Note that for $\epsilon_1$-typical sequences $\underline{\mathbf{s}}$, the following inequality holds:
\begin{equation*}
\frac{1}{n}\log_2 |\mathcal{L}_1(\underline{\mathbf{s}})| \leq H(\underline{V}, \underline{X}, Y|\underline{S})+ H(S_1,V_1,X_1, S_3, V_3, X_3| Y S_2 V_2 X_2)+\delta_5(\epsilon),
\end{equation*}
where $\delta_5(\epsilon)\rightarrow 0$ as $\epsilon \rightarrow 0$. Note 
\begin{align}\nonumber
H(\underline{V},\underline{X}, Y| \underline{S})&=H(\underline{V}|\underline{S})+H(\underline{X}|\underline{S}, \underline{V})+H(Y|\underline{X})\\\label{eq:H(VXY|S)}
&=2\log_2q+\sum_{i=1}^3H(X_i|S_i,V_i)+H(Y|\underline{X}),
\end{align}
where the first equality holds by chain rule and the Markov chain $(\underline{S}, \underline{V})\leftrightarrow \underline{X}\leftrightarrow Y$. The second equality holds, because, from \eqref{eq: thm mac corr scr joint dist}, $\underline{V}$ are independent of the other random variables and $P_{V_1V_2V_3}=\frac{1}{q^2}\11\{V_3=V_1\+_q V_2\}$. Therefore, $p_{e_1}\rightarrow 0$ as $n\rightarrow \infty$,  if 
\begin{align}\label{eq: initial_inequality_case 1}
H(S_1,V_1,X_1, S_3, V_3, X_3| Y S_2 V_2 X_2)\leq \log_2q+ H(X_1|S_1V_1)+H(X_3|S_3V_3)
\end{align}
Next, we simplify the right-hand side terms in \eqref{eq: initial_inequality_case 1}. From \eqref{eq: thm mac corr scr joint dist}, the Markov chain $(S_{i^c}, V_{i^c}, X_{i^c})\leftrightarrow S_i \leftrightarrow X_i$ holds for all $i\in \{1,2,3\}$, where $i^c\deq \{1,2,3\}/\{i\}$.  Therefore, the right-hand side above equals to 
\begin{equation}\label{eq:log q right HS}
\log_2 q+ H(X_1X_3|\underline{S}, V_1 V_3 X_2 V_2)= H(X_1X_3V_1V_3|\underline{S} X_2V_2),
\end{equation}
where the equality holds by chain rule and the following argument:
\[
H(V_1V_3|\underline{S}X_2V_2)=H(V_1|\underline{S}X_2V_2)=H(V_1|\underline{S}V_2)=H(V_1|V_2)=H(V_1)=\log_2 q.
\]

 We simplify the left-hand side in \eqref{eq: initial_inequality_case 1}. Using  chain rule  
\begin{align*}
H(S_1,V_1,X_1, S_3, V_3, X_3| Y S_2 V_2 X_2)&=H(V_1,X_1, V_3, X_3| Y S_2 V_2 X_2)+H(S_1 S_3|Y S_2 \underline{V}~\underline{X})\\
 &= H(V_1,X_1, V_3, X_3| Y S_2 V_2 X_2)+H(S_1|S_2 \underline{V}~\underline{X}),
\end{align*}
where the second equality holds due to the Markov chain $\underline{S}\leftrightarrow \underline{X}\leftrightarrow Y$ and the assumption that $S_1\+_q S_2 \+_q S_3=0$. Note that
\begin{align*}
H(S_1|S_2 \underline{V}~\underline{X})&=H(S_1|S_2X_2V_2)-I(S_1;X_1V_1X_3V_3|S_2V_2X_2)\\
&=H(S_1|S_2)-I(S_1;X_1V_1X_3V_3|S_2V_2X_2),
\end{align*}
where, the last equality holds because $V_2$ is independent of $S_1$ and $X_2$ is a function of $(S_2, V_2)$. 
Therefore, using the above arguments, the inequality in (\ref{eq: initial_inequality_case 1}) is simplified to
\begin{align*}
H(S_1|S_2)&\leq I(S_1;X_1V_1X_3V_3|S_2V_2X_2)-H(V_1,X_1, V_3, X_3| Y S_2 V_2 X_2)+H(X_1X_3V_1V_3|\underline{S} X_2V_2)\\
&=I(X_1V_1X_3V_3;Y|S_2V_2X_2)=I(X_1X_3;Y|S_2V_2X_2).
\end{align*}

As a result, $p_{e_1}(\underline{\mathbf{s}})$ can be made sufficiently small for large enough $n$, if the inequality  $$H(S_1|S_2)\leq I(X_1X_3;Y|S_2V_2X_2)$$ is satisfied.

\noindent{\bf Case 2, $\mathbf{\tilde{s}}_1 = \mathbf{s}_1,  \mathbf{\tilde{s}}_2 \neq  \mathbf{s}_2$:}

\hspace{10pt} This case corresponds to  $p_{e_2}(\underline{\mathbf{s}})$ which is defined using a similar expression as for  $p_{e_1}(\underline{\mathbf{s}})$; but with the conditions in the second summation replaced with $\underline{\tilde{\mathbf{s}}} \neq \underline{\mathbf{s}},  \mathbf{\tilde{s}}_1= \mathbf{s}_1,  \mathbf{\tilde{v}}_1= \mathbf{v}_1$. Therefore, we have 
\begin{align*}
\mathcal{L}_2(\underline{\mathbf{s}})\deq \Big\{ (\underline{\mathbf{v}}, \underline{\mathbf{x}},\mathbf{y}, \underline{\mathbf{\tilde{s}}},\underline{\tilde{\mathbf{v}}},\underline{\tilde{\mathbf{x}}})\in \mathcal{L}(\underline{\mathbf{s}}): \mathbf{\tilde{s}}_1 = \mathbf{s}_1,  \mathbf{\tilde{s}}_2 \neq \mathbf{s}_2, \mathbf{\tilde{v}}_1= \mathbf{v}_1,  \mathbf{\tilde{x}}_1= \mathbf{x}_1\Big\}.
\end{align*}
By symmetry and using a similar argument as in the first case, we can show that $p_{e_2}(\underline{\mathbf{s}})\rightarrow 0$ as $n\rightarrow \infty$ if the following inequality holds $$H(S_2|S_1) \leq I(X_2X_3;Y|S_1V_1X_1).$$

\noindent{\bf Case 3, $\mathbf{\tilde{s}}_1 \neq \mathbf{s}_1,  \mathbf{\tilde{s}}_2 \neq  \mathbf{s}_2, \mathbf{\tilde{s}}_1 \oplus_q \mathbf{\tilde{s}}_2=\mathbf{s}_1\oplus_q \mathbf{s}_2 $:}

\hspace{10pt} In this case
\begin{align*}
&P\{\mathbf{v}_j=\mathbf{V}_j(\mathbf{s}_j), \tilde{\mathbf{v}}_j=\mathbf{V}_j(\tilde{\mathbf{s}}_j)~ j=1,2\} = q^{-3n} \mathbbm{1}\{ \mathbf{\tilde{v}}_1 \oplus_q \mathbf{\tilde{v}}_2=\mathbf{v}_1\oplus_q \mathbf{v}_2 \}\\
&P\{\mathbf{x}_l=\mathbf{X}_l(\mathbf{s}_l, \mathbf{v}_l), \tilde{\mathbf{x}}_l=\mathbf{X}_l(\tilde{\mathbf{s}}_l, \tilde{\mathbf{v}}_l), l=1,2,3\}\\
&\hspace{140pt}\leq 2^{-n[2H(X_1|S_1V_1)+2H(X_2|S_2V_2)+H(X_3|S_3V_3)-\delta_6(\epsilon)]} \mathbbm{1}\{\mathbf{\tilde{x}}_3=\mathbf{x}_3\}.
\end{align*}
By assumption $\mathbf{s}_1 \+_q \mathbf{s}_2 \+_q \mathbf{s}_3=0$ and $\mathbf{v}_1 \+_q \mathbf{v}_2 \+_q \mathbf{v}_3=0$. Therefore, the first probability is nonzero only when $\mathbf{\tilde{v}}_3=\mathbf{v}_3$. Hence, as $\mathbf{s}_3=\mathbf{\tilde{s}}_3$, we get $X_3(\tilde{\mathbf{s}}_3, \tilde{\mathbf{v}}_3)=X_3(\mathbf{s}_3, \mathbf{v}_3)$. As a result, we can define
\begin{align*}
\mathcal{L}_3(\underline{\mathbf{s}})\deq \Big\{ (\underline{\mathbf{v}}, \underline{\mathbf{x}},\mathbf{y}, \underline{\mathbf{\tilde{s}}},\underline{\tilde{\mathbf{v}}},\underline{\tilde{\mathbf{x}}})\in \mathcal{L}(\underline{\mathbf{s}}):& \mathbf{\tilde{s}}_1 \neq \mathbf{s}_1,  \mathbf{\tilde{s}}_2\neq  \mathbf{s}_2,\\
&\mathbf{\tilde{s}}_1 \oplus_q \mathbf{\tilde{s}}_2=\mathbf{s}_1\oplus_q \mathbf{s}_2, \mathbf{\tilde{v}}_1 \oplus_q \mathbf{\tilde{v}}_2=\mathbf{v}_1\oplus_q \mathbf{v}_2, \mathbf{\tilde{x}}_3=\mathbf{x}_3\Big\}.
\end{align*}
As a result,  the contribution of this case ($p_{e_3}$) is bounded by 
\begin{align*}
p_{e_3}(\underline{\mathbf{s}})&\leq |\mathcal{L}_3(\underline{\mathbf{s}})| 2^{-nH(Y|\underline{X})}\frac{1}{q^{3n}} 2^{-n[2H(X_1|S_1V_1)+2H(X_2|S_2V_2)+ H(X_3|S_3V_3)-\delta_7(\epsilon)]},
\end{align*}

Note that for $\epsilon_1$-typical $\underline{\mathbf{s}}$, we have 
\begin{equation*}
\frac{1}{n}\log_2 |\mathcal{L}_3(\underline{\mathbf{s}})| \leq H(\underline{V}, \underline{X}, Y|\underline{S})+ H(S_1,V_1,X_1, S_2, V_2, X_2| Y S_3 V_3 X_3)+\delta_8(\epsilon).
\end{equation*}
Therefore, from \eqref{eq:H(VXY|S)} and the above inequality, 
$p_{e_3}(\underline{\mathbf{s}}) \rightarrow 0$, if 
\begin{align*}
H(S_1,V_1,X_1, S_2, V_2, X_2| Y S_3 V_3 X_3)&\leq \log_2 q+H(X_1|S_1V_1)+H(X_2|S_2V_2)\\
 &= H(X_1,X_2,V_1,V_2|S_1S_2S_3V_3X_3),
\end{align*}
where the inequality above holds using a similar argument applied in \eqref{eq:log q right HS}. By symmetry and using a similar argument as in the first case, this inequality is equivalent to $$ H(S_1S_2|S_3) \leq I(X_1, X_2; Y| S_3V_3X_3).$$

{
\noindent{{\bf Case 4, $\mathbf{\tilde{s}}_1\+a \mathbf{\tilde{s}}_2 = \mathbf{s}_1\+a \mathbf{s}_2$, $a\in \FF_q/\{0,1\}$:}}

\hspace{10pt} From Lemma \ref{lem: prob of sG}, 
\begin{align*}
P\{\mathbf{v}_j=\mathbf{V}_j(\mathbf{s}_j), \tilde{\mathbf{v}}_j=\mathbf{V}_j(\tilde{\mathbf{s}}_j)~ j=1,2\} = q^{-3n} \mathbbm{1}\{ \mathbf{\tilde{v}}_1 \oplus_q a \mathbf{\tilde{v}}_2=\mathbf{v}_1\oplus_q a \mathbf{v}_2 \}.
\end{align*}
Therefore, the error probability in this case, i.e., $p_{e_4}(\underline{\mathbf{s}})$ satisfies
\begin{align*}
p_{e_4} (\underline{\mathbf{s}}) \deq &\sum_{a=1}^{q-1}  \sum_{\mathcal{L}_4(a, \underline{\mathbf{s}})}  P^n_{Y|\underline{X}}(\mathbf{y}|\underline{\mathbf{x}})q^{-3n}  P\{\mathbf{x}_l=\mathbf{X}_l(\mathbf{s}_l, \mathbf{v}_l), \tilde{\mathbf{x}}_l=\mathbf{X}_l(\tilde{\mathbf{s}}_l, \tilde{\mathbf{v}}_l) ~  l=1,2,3\},
\end{align*}
where  
\begin{align*}
\mathcal{L}_4(a,\underline{\mathbf{s}})\deq \Big\{ (\underline{\mathbf{v}}, \underline{\mathbf{x}},\mathbf{y}, \underline{\mathbf{\tilde{s}}},\underline{\tilde{\mathbf{v}}},\underline{\tilde{\mathbf{x}}})\in \mathcal{L}(\underline{\mathbf{s}}): ~ \mathbf{\tilde{s}}_1\+a \mathbf{\tilde{s}}_2 = \mathbf{s}_1\+a \mathbf{s}_2,~ \mathbf{\tilde{v}}_1\+a \mathbf{\tilde{v}}_2 = \mathbf{v}_1\+a \mathbf{v}_2\Big\}.
\end{align*}
Also, observe that 
\begin{align*}
P\{\mathbf{x}_l=\mathbf{X}_l(\mathbf{s}_l, \mathbf{v}_l), \tilde{\mathbf{x}}_l=\mathbf{X}_l(\tilde{\mathbf{s}}_l, \tilde{\mathbf{v}}_l) ~  l=1,2,3\}\leq 2^{-2n[\sum_{i=1}^3H(X_i|S_iV_i)]-\delta_9(\epsilon)}.
\end{align*}
where $\delta_9(\cdot)$ is a continuous function of $\epsilon$ with $\lim_{\epsilon\rightarrow 0}\delta_9(\epsilon)=0$. Consequently, for any typical sequences $\underline{\mathbf{s}}$,  the following upper bound holds: 
\begin{align*}
p_{e_4}(\underline{\mathbf{s}}) \leq \sum_{a=1}^{q-1} |\mathcal{L}_4(a,\underline{\mathbf{s}})| 2^{-nH(Y|\underline{X})}q^{-3n} 2^{-2n[\sum_{i=1}^3H(X_i|S_iV_i)]}2^{n\delta_{10}(\epsilon)}.
\end{align*} 

Note that for any non-zero $a\in \FF_q$ and any typical sequence $\underline{\mathbf{s}}$, the cardinality of $\mathcal{L}_4$ satisfies the inequality
\begin{equation*}
\frac{1}{n}\log_2 |\mathcal{L}_4(a,\underline{\mathbf{s}})|\leq H(\underline{V}, \underline{X},Y|\underline{{S}})+H(\underline{S},\underline{V},\underline{X}|Y, S_1\+_q aS_2, V_1\+_q aV_2 )+\delta_{11}(\epsilon).
\end{equation*}

Note that $$H(\underline{V},\underline{X}, Y| \underline{S})=2\log_2q+\sum_{i=1}^3H(X_i|S_i,V_i)+H(Y|\underline{X}).$$
Therefore, from the above inequalities, $p_{e_4}(\underline{\mathbf{s}}) \rightarrow 0$ as $n\rightarrow \infty$, if 
\begin{equation}\label{ineq: case 4}
H(\underline{S},\underline{V},\underline{X}|Y, S_1\+_q aS_2, V_1\+_q aV_2 )< \log q+ \sum_{i=1}^3H(X_i|S_i,V_i)
\end{equation}
From the joint probability distribution given in \eqref{eq: thm mac corr scr joint dist}, conditioned on $(\underline{S}, \underline{V})$ the random variables $(X_1, X_2, X_3)$ are mutually independent. Hence, $\sum_{i=1}^3H(X_i|S_i,V_i)=H(\underline{X}|\underline{S}, \underline{V})$ and the right-hand side of the above inequality simplifies to $\log q + H(\underline{X}|\underline{S}, \underline{V})$. 
Next, we simplify the left-hand side of the above inequality. For that we have
\begin{align*}
H(\underline{S},\underline{V},\underline{X}|& Y, S_1\+_q aS_2, V_1\+_q aV_2 )=H(\underline{V},\underline{X}|Y, S_1\+_q aS_2, V_1\+_q aV_2 )+H(\underline{S}|S_1\+_q aS_2, \underline{X}, \underline{V})\\
&=H(\underline{S}|S_1\+_q aS_2, V_1\+_q aV_2 )-I(\underline{X}, \underline{V}; Y|S_1\+_q aS_2,  V_1\+_q aV_2 )\\
&+H(\underline{X},\underline{V}|\underline{S}, S_1\+_q aS_2, V_1\+_q aV_2 )\\
&=H(\underline{S}|S_1\+_q aS_2 )-I(\underline{X}, \underline{V}; Y|S_1\+_q aS_2,  V_1\+_q aV_2 )+H(\underline{X},\underline{V}|\underline{S}, V_1\+_q aV_2 )
\end{align*}
where the first equality holds by chain rule and the Markov chain $\underline{S}\leftrightarrow \underline{X}\leftrightarrow Y$. The second equality holds by the definition of the mutual information. The last equality holds as $(V_1, V_2, V_3)$ are independent of $(S_1,S_2,S_3)$. As a result of the above argument, the inequality in \eqref{ineq: case 4} is equivalent to the following inequality:
\begin{align*}
H(\underline{S}|S_1\+_q aS_2) &< I(\underline{X}, \underline{V};Y|S_1\+_q aS_2, V_1\+_q aV_2 )-H(\underline{X},\underline{V}|\underline{S}, V_1\+_q aV_2 )+\log q+H(\underline{X}|\underline{S}, \underline{V})\\
&=I(\underline{X}, \underline{V};Y|S_1\+_q aS_2, V_1\+_q aV_2 )-H(\underline{V}|V_1\+_q aV_2 )\\
&~~~-H(\underline{X}|\underline{S}, \underline{V}, V_1\+_q aV_2 )+\log q+H(\underline{X}|\underline{S}, \underline{V})\\
&=I(\underline{X}, \underline{V};Y|S_1\+_q aS_2, V_1\+_q aV_2 )-H(\underline{V}|V_1\+_q aV_2 )+\log q,
\end{align*}
where the first equality holds by the chain rule and the fact that $\underline{V}$ is independent of $\underline{S}$. In what follows, we show that the last two terms above cancel each other.  Since $V_1$ and $V_2$ are independent random variables with uniform distribution over $\FF_q$, then so is $V_1$ and $V_1\+_q aV_2$ for any $a \in \FF_1/\{0\}$. Therefore, as $V_3 \+_q V_1\+_q V_2=0$ we have 
\begin{align*}
H(\underline{V}|V_1\+_q aV_2 )&=H(V_1,V_2|V_1\+_q aV_2 )= H(V_1,V_1\+_q aV_2 |V_1\+_q aV_2 )\\
&=H(V_1 | V_1\+_q aV_2 )=\log q.
\end{align*}

As a result, we showed that $p_{e_4}(\underline{\mathbf{s}}) \rightarrow 0$ as $n\rightarrow \infty$, if 
$$H(\underline{S}|S_1\+_q aS_2) \leq I(\underline{X}, \underline{V};Y|S_1\+_q aS_2, V_1\+_q aV_2 ).$$ }

\noindent{\bf Case 5}, $\mathbf{\tilde{s}}_i \neq \mathbf{s}_i,  i=1,2,3$ and  $\mathbf{\tilde{s}}_1\+a\mathbf{\tilde{s}}_2\neq \mathbf{{s}}_1\+a\mathbf{{s}}_2 $ for all $a\in \FF_q$:\\
\hspace{10pt} Observe that,
\begin{align*}
&\mathcal{L}_5(\underline{\mathbf{s}})\deq \Big\{ (\underline{\mathbf{v}}, \underline{\mathbf{x}},\mathbf{y}, \underline{\mathbf{\tilde{s}}},\underline{\tilde{\mathbf{v}}},\underline{\tilde{\mathbf{x}}})\in \mathcal{L}(\underline{\mathbf{s}}):  \mathbf{\tilde{s}}_1 \neq \mathbf{s}_1,  \mathbf{\tilde{s}}_2\neq  \mathbf{s}_2, \mathbf{\tilde{s}}_1\+a\mathbf{\tilde{s}}_2\neq \mathbf{{s}}_1\+a\mathbf{{s}}_2, \forall a\in \FF_q\Big\}\\
&P\{\mathbf{v}_j=\mathbf{V}_j(\mathbf{s}_j), \tilde{\mathbf{v}}_j=\mathbf{V}_j(\tilde{\mathbf{s}}_j)~ j=1,2\} = q^{-4n}\\
&P\{\mathbf{x}_l=\mathbf{X}_l(\mathbf{s}_l, \mathbf{v}_l), \tilde{\mathbf{x}}_l=\mathbf{X}_l(\tilde{\mathbf{s}}_l, \tilde{\mathbf{v}}_l) ~  l=1,2,3\}\leq 2^{-2n[\sum_{l=1}^3 H(X_l|S_lV_l)-\delta_9(\epsilon)]}.
\end{align*}
Therefore, the contribution of this case is simplified to $p_{e_5}(\underline{\mathbf{s}}) \approx q^{-2n}  2^{nH(  \underline{S}, \underline{V}, \underline{X}| Y)}  2^{-n\sum_{l=1}^3 H(X_l|S_lV_l)}$. As a result, one can show that $P_{e_5}\rightarrow 0$, if $H(S_1S_2S_3) \leq I(X_1X_2X_3; Y)$.

Finally, note that $P_e(\underline{\mathbf{s}}) = \sum_{i=1}^5 P_{ei}(\underline{\mathbf{s}}).$ Moreover, $P_{ei}(\underline{\mathbf{s}})$ depends on $\underline{\mathbf{s}}$ only through its PMF. Therefore, for any typical $\underline{\mathbf{s}}$,  $P_e$ approaches zero as $n\rightarrow \infty$, if the following bounds are satisfied:
{ \begin{align*}
H(S_1|S_2) &\leq I(X_1 X_3;Y| S_2 V_2 X_2)\\
H(S_2|S_1) &\leq I(X_2 X_3;Y| S_1 V_1 X_1)\\
H(S_1S_2| S_1 \oplus_q S_2) & \leq I(X_1 X_2;Y| S_1 \oplus_q S_2,  V_3 X_3)\\
H(S_1S_2| S_1 \oplus_q a S_2) & \leq I(X_1, X_2, X_3;Y| S_1 \oplus_q a S_2,  V_1\+_q a V_2)\\
H(S_1,S_2) &\leq I(X_1 X_2 X_3;Y).
\end{align*} }

\end{proof}

\section{Proof of Theorem \ref{thm:CES is subopt} }\label{appx: proof_ thm_CES is sub opt}

\begin{lem}\label{lem:MAC in example is unstructured}
For the MAC in Example \ref{ex: CES is suboptimal},  $I(X_1, X_2, X_3; Y)\leq 2-H(N)$, with equality if and only if  $X_3=X_1\oplus_2 X_2$ with probability one, and $X_3$ is uniform over $\{0,1\}$. 
\end{lem}

\begin{proof}
Note  $I(X_1,X_2,X_3;Y)=H(Y)-H(N)$. We proceed by finding all the necessary and sufficient conditions on $P_{X_1,X_2,X_3}$ for which $Y$ is uniform over $\ZZ_4$. From Figure \ref{fig: Exp1}, $Y= (X_1\oplus_2 X_2) \oplus_4 X_3 \oplus_4 N$. Denote $X'_2=X_1\oplus_2 X_2$. Let $P(X'_2 \oplus_4 X_3=i)=q(i)$ where $i=1,2,3,4$. Since $X'_2$ and $X_3$ are binary, $q(3)=0$. Given the distribution of $N$ is Table \ref{tab: N},  the distribution of $Y$ is as follows:
\begin{subequations}\label{eqs: P(Y)}
\begin{align}
P(Y=0)&=q(0)(\frac{1}{2}-\delta)+q(2)\delta,\\
P(Y=1)&=q(0)\frac{1}{2}+q(1)(\frac{1}{2}-\delta),\\
P(Y=2)&=q(0)\delta+q(1)\frac{1}{2}+q(2) (\frac{1}{2}-\delta),\\
P(Y=3)&=q(2)\frac{1}{2}+q(1)\delta.
\end{align}
\end{subequations}
It's not difficult to check that the only solution for the equations in \eqref{eqs: P(Y)} is 
\begin{equation*}
q(0)=q(2)=\frac{1}{2}, \quad q(1)=0.
\end{equation*} 
Note that by definition 
 \begin{align*}
 q(1)=P(X'_2=0,X_3=1)+P(X'_2=1,X_3=0).
 \end{align*}
 Therefore, $q(1)=0$ implies that  $X_3=X'_2$ with probability one. If this condition is satisfied, then $q(0)=P(X_3=0)$ and $q(2)=P(X_3=1)$. Since $q(0)=q(2)=\frac{1}{2}$ then $X_3$ is uniform over $\{0,1\}$.
To sum up, we proved that $Y$ is uniform, if and only if 1) $X_3=X_1\oplus_2 X_2$. 2) $X_3$ is uniform over $\{0,1\}$.
\end{proof}

\begin{lem}\label{lem:total variation dist}
Let $\mathscr{P}_1$ be the set of all distributions $P_{X_1,X_2,X_3}^*$ that satisfies the conditions in Lemma \ref{lem:MAC in example is unstructured}. Let $\mathscr{P}_2$ be the set of all distributions $P_{X_1, X_2, X_3}$ which is the marginal of  $P_{S_1,S_2,S_3}P_{X_1,X_2,X_3|S_1,S_2,S_3}$ for some source triplet $(S_1,S_2,S_3)$ in Example \ref{ex: CES is suboptimal} with parameters $\sigma \in (0, \frac{1}{2}], \gamma\in (0, \gamma^*]$ and conditional distribution of the form $P_{X_1,X_2,X_3|S_1,S_2,S_3}=\prod_{i=1}^3 P_{X_i|S_i}$. Then the total variation distance between $\mathscr{P}_1$ and $\mathscr{P}_2$ satisfies
\begin{equation*}
TV(\mathscr{P}_1, \mathscr{P}_2)\geq  \frac{1}{6}-\frac{\gamma^*}{3}.
\end{equation*}

Moreover, there exists $\alpha(\gamma^*) >0$ such that $I(X_1,X_2,X_3; Y) \leq 2-H(N)-\alpha(\gamma^*)$ for all $P_{X_1,X_2,X_3}\in \mathscr{P}_2$.
\end{lem}
\begin{proof}
 Let $\overline{\gamma^*}\deq 1-\gamma^*$ and assume for some $\epsilon \geq 0$ there exist sources with parameters $\sigma_{\epsilon}\in (0, \frac{1}{2}]$ and $\gamma_{\epsilon}\in (0, \gamma^*]$ and conditional distributions  $P^{\epsilon}_{X_i|S_i}, i=1,2,3$ and a distribution $P_{X_1,X_2,X_3}^*$ satisfying the conditions in Lemma \ref{lem:MAC in example is unstructured}  such that total variation distance between the resulted PMF $P^{\epsilon}_{X_1,X_2,X_3}$ and $P_{X_1,X_2,X_3}^*$ is equal to $\epsilon$. Then for  $P^{\epsilon}_{X_1,X_2,X_3}$ the following inequalities hold:
\begin{align}\label{eq:cond for P_eps}
P^{\epsilon}(X_3\neq X_1\+X_2)\leq \epsilon, \quad  \text{and} \quad   \Big|P^{\epsilon}(X_3=1)-\frac{1}{2}\Big|\leq  \epsilon.
\end{align}
The second inequality implies 
\begin{subequations}
\begin{align}\label{subeq:P_eps 1}
\gamma_{\epsilon} ~P^{\epsilon}_{X_3|S_3}(1|1)+\overline{\gamma_{\epsilon}} ~P^{\epsilon}_{X_3|S_3}(1|0)\in [\frac{1}{2}-\epsilon, \frac{1}{2}+\epsilon],\\\label{subeq:P_eps 2}
\gamma_{\epsilon} ~P^{\epsilon}_{X_3|S_3}(0|1)+\overline{\gamma_{\epsilon}} ~P^{\epsilon}_{X_3|S_3}(0|0)\in [\frac{1}{2}-\epsilon, \frac{1}{2}+\epsilon].
\end{align}
\end{subequations}
Since the first terms in \eqref{subeq:P_eps 1} and \eqref{subeq:P_eps 2} are non-negative and $\gamma_\epsilon\leq \gamma^*$, then
\begin{align*}
\frac{1}{2}+\epsilon \geq & ~ \overline{\gamma_{\epsilon}} ~P^{\epsilon}_{X_3|S_3}(1|0) \geq  \overline{\gamma^*} ~P^{\epsilon}_{X_3|S_3}(1|0),\\
\frac{1}{2}+\epsilon \geq & ~ \overline{\gamma_{\epsilon}} ~P^{\epsilon}_{X_3|S_3}(0|0) \geq  \overline{\gamma^*} ~P^{\epsilon}_{X_3|S_3}(0|0).
\end{align*}
Since $P^{\epsilon}_{X_3|S_3}(0|0)+P^{\epsilon}_{X_3|S_3}(1|0)=1$, then the above inequalities imply the following
\begin{subequations}\label{subeq:P_eps}
\begin{align}
     \frac{1}{2}+\epsilon   &\geq  \overline{\gamma^*} ~P^{\epsilon}_{X_3|S_3}(1|0) \geq \overline{\gamma^*}-\frac{1}{2}-\epsilon\\
          \frac{1}{2}+\epsilon   &\geq  \overline{\gamma^*} ~P^{\epsilon}_{X_3|S_3}(0|0) \geq \overline{\gamma^*}-\frac{1}{2}-\epsilon
\end{align}
\end{subequations}
From the law of total probability, the first condition in \eqref{eq:cond for P_eps} is equivalent to 
\begin{align*}
\sum_{\underline{s}} \sum_{x_1, x_2} P^{\epsilon}_{\underline{S}}(\underline{s})  P^{\epsilon}_{X_1|S_1}(x_1|s_1)P^{\epsilon}_{X_2|S_2}(x_2|s_2)P^{\epsilon}_{X_3|S_3}(\overline{x_1\+x_2}|s_3)\leq \epsilon,
\end{align*}
where $P^{\epsilon}_{\underline{S}}$ is the joint PMF of the sources with parameters $\sigma_{\epsilon}, \gamma_{\epsilon}$, and $\overline{x_1\+_2 x_2}\deq 1\+_2 x_1\+_2x_2 $. By considering the case $s_1=s_2=s_3=0$, the above inequality implies
\begin{align*}
\epsilon &\geq \sum_{x_1, x_2} \overline{\gamma_{\epsilon}}~ \overline{\sigma_{\epsilon}} ~ P^{\epsilon}_{X_1|S_1}(x_1|0)P^{\epsilon}_{X_2|S_2}(x_2|0)P^{\epsilon}_{X_3|0}(\overline{x_1\+_2 x_2}|0)\\
&\geq \sum_{x_1, x_2}  \overline{\gamma^*}~ \frac{1}{2}~ P^{\epsilon}_{X_1|S_1}(x_1|0)P^{\epsilon}_{X_2|S_2}(x_2|0)P^{\epsilon}_{X_3|0}(\overline{x_1\+_2 x_2}|0)\\
&\geq \sum_{x_1, x_2}   \frac{1}{2}~ (\overline{\gamma^*}-\frac{1}{2}-\epsilon) P^{\epsilon}_{X_1|S_1}(x_1|0)P^{\epsilon}_{X_2|S_2}(x_2|0)=\frac{1}{2}~ (\overline{\gamma^*}-\frac{1}{2}-\epsilon),
\end{align*}
where the third inequality holds from the bounds in \eqref{subeq:P_eps}. As a result, these inequalities imply that $\epsilon \geq \frac{1}{3}(\overline{\gamma^*}-\frac{1}{2})$. From Lemma \ref{lem:MAC in example is unstructured} and the continuity of the mutual information in total variation distance \cite{Csiszar}, the second statement of the lemma follows. 

\end{proof}

\begin{lem}\label{lem: neigh of gamma* is not achievable by CES}
For the setup in Example \ref{ex: CES is suboptimal}, there exists $\epsilon >0$ such that any source triple $(S_1,S_2,S_3)$ with parameters  $(\sigma >0, \gamma\geq \gamma^* - \epsilon)$ does not satisfy the sufficient conditions stated in Proposition \ref{prep: CES_three_user}. 
\end{lem}
\begin{proof}
We prove the lemma by a contradiction. Suppose $\forall \epsilon>0$ there exist $\sigma>0$ and $\gamma\geq \gamma^*-\epsilon$ such that the sufficient conditions in Proposition \ref{prep: CES_three_user} are satisfied. Consider the fourth inequality in Proposition  \ref{prep: CES_three_user}.  Since $\sigma>0$ there is no common part. Let $U'=U_{123}U_{12}U_{13}U_{23}$. Then,  the following holds
\begin{equation}\label{eq: CES outer bound 1}
h(\gamma)+h(\sigma) \leq \max_{p(u')p(\underline{x}|u'\underline{s})} I(X_1X_2 X_3;Y|U'),
\end{equation}
where $$p(\underline{s},\underline{x}, u')=p(\underline{s})p(u')p(x_1|s_1, u')p(x_2|s_2,u')p(x_3|s_3,u').$$
Since $U'$ is independent of the sources, and appears in the conditioning in the mutual information term, the inequality in (\ref{eq: CES outer bound 1}) is equivalent to 
\begin{equation}\label{eq: CES outter bound 2}
h(\gamma)+h(\sigma) \leq \max_{p(\underline{x}|\underline{s})} I(X_1X_2 X_3;Y),
\end{equation}
where $p(\underline{s},\underline{x})=p(\underline{s})p(x_1|s_1)p(x_2|s_2)p(x_3|s_3).$
From Lemma \ref{lem:total variation dist}, the right-hand side in (\ref{eq: CES outter bound 2}) is less than $2-H(N)-\alpha$, for some $\alpha>0$ (which depends only on $\gamma^*$ which is a function of $\delta$). As $h(\gamma^*)=2-H(N)$, by the bound above, $h(\gamma)+h(\sigma)\leq h(\gamma^*)-\alpha$. Thus, as $h(\sigma)>0$, we get $h(\gamma)<h(\gamma^*)-\alpha$. By the continuity and monotonicity of  the binary entropy function, $\gamma < h^{-1}(h(\gamma^*)-\alpha)=
\gamma^* - \lambda(\alpha)$, where $\lambda(\alpha)>0$. Hence, as $\gamma\geq \gamma^*-\epsilon$, then $\epsilon$ must be greater than $\lambda(\alpha)$ which is a contradiction.
%
%
%
\end{proof}
\begin{lem}\label{lem: neigh of gamma is achievable}
There exists a non-negative function $\sigma_0(\gamma)$ such that 1) $\sigma_0(\gamma)>0$ for all $\gamma\in [0,\gamma^*)$, and 2)  any source with parameters $0\leq \gamma \leq \gamma^*, 0\leq\sigma \leq \sigma_0(\gamma)$ is transmissible.
\end{lem}
\begin{proof}
 For the setup in Example \ref{ex: CES is suboptimal}, the bounds given in Theorem \ref{them: achievable-rate-for -proposed scheme} are simplified to 
 \begin{subequations}\label{eq: transmittable sources}
\begin{align}\label{eq: transmittable sources bound 1}
h(\gamma) &\leq I(X_2 X_3;Y| X_1 S_1V_1)\\\label{eq: transmittable sources bound 2}
h(\sigma) &\leq I(X_1 X_2; Y|X_3 S_3 V_3)\\\label{eq: transmittable sources bound 3}
h(\gamma)+h(\sigma)-h(\sigma * \gamma)&\leq I(X_1 X_3;Y| X_2 S_2 V_2)\\\label{eq: transmittable sources bound 4}
h(\gamma)+h(\sigma)&\leq I(X_1 X_2 X_3;Y).
\end{align}
\end{subequations}


Let $E_1\sim Ber(\alpha)$, and set $X_1=V_1\+E_1$ and $X_2=V_2, X_3=V_3$, where $(V_1,V_2,V_3)$ are as in Theorem \ref{them: achievable-rate-for -proposed scheme}; that is they are pairwise independent Bernoulli random variables with joint PMF $P_{V_1,V_2,V_3}=\frac{1}{4}\11\{V_3=V_1\+_2V_2\}$. Next, using these random variables, we further simplify the conditions in  \eqref{eq: transmittable sources}.

We start by the first condition given in \eqref{eq: transmittable sources bound 1}. The right-hand side is simplified to
\begin{align}\nonumber
I(X_2 X_3;Y| X_1 S_1V_1)&=H((X_1\+_2 X_2)\+_4X_3\+_4N|X_1 V_1)-H(N)\\\nonumber
&=H((E_1\+_2V_1\+_2V_2)\+_4(V_1\+_2V_2)\+_4N|E_1, V_1)-H(N)\\
&=P(E_1=0) [ H((V_1\+_2V_2)\+_4(V_1\+_2V_2)\+_4N|V_1)-H(N)]\\\label{eq:transmittable sources bound 1 simpl}
&=(1-\alpha)(2-H(N)),
\end{align}
where the first equality holds as $Y=(X_1\+_2X_2)\+_4X_3\+_4N$ and $X_i, i=1,2,3$ are independent of the sources. The fourth equality holds as $H(X\+_4X\+_4N)=2$ and $H((1\+_2X)\+_4X\+_4N)=H(N)$  when $X$ is uniform over $\{0,1\}$. Therefore, from \eqref{eq:transmittable sources bound 1 simpl}, the first condition gives 
$h_b(\gamma)\leq 2-H(N)$. This condition is always satisfied for any $\gamma \leq \gamma^*$. This is due to the monotonicity of the binary entropy function.

Next, we evaluate the second condition given by \eqref{eq: transmittable sources bound 2}.  Using a similar argument, the right-hand side of \eqref{eq: transmittable sources bound 2} is simplified to
\begin{equation}\label{eq:transmittable sources bound 2 simpl}
I(X_1 X_2; Y|X_3 S_3 V_3)=H((X_1\+_2X_2)\+_4N|X_3 V_3)-H(N)=H(E_1\+_4N)-H(N).
\end{equation}
Hence, the second condition gives $h_b(\sigma)\leq \eta_1(\alpha)$, where $\eta_1(\alpha)\deq H(E_1\+_4N)- H(N)$. We show that $\eta_1(\alpha)$ is strictly positive for all $\alpha\in (0,\frac{1}{2}]$. For that we have $H(N)=1+\frac{1}{2}h_b(2\delta)$ and 
\begin{align*}
H(E_1\+_4N)&=1+\frac{1}{2}[h_b(2\alpha\delta)+h_b(2(1-\alpha)\delta+\alpha)]\\
&\geq 1+\frac{1}{2}[h_b(2\alpha\delta)+(1-\alpha)h_b(2\delta)],
\end{align*}
where the first inequality holds due to the convexity of binary entropy function and the fact that $h_b(1)=0$. 
Hence, $\eta_1(\alpha)\geq\frac{1}{2}[h_b(2\alpha\delta)-\alpha h_b(2\delta)]$. When $\delta \in (0, \frac{1}{4}]$, the equality $h_b(2\alpha\delta)=\alpha h_b(2 \delta)$  holds if and only if $\alpha\in \{0,1\}$. As a result of this and due to the convexity of binary entropy, the strict inequality $h_b(2\alpha\delta)>\alpha h_b(\delta)$ holds.

For the third and fourth conditions, the right-hand sides of \eqref{eq: transmittable sources bound 3} and \eqref{eq: transmittable sources bound 4} are simplified to
\begin{align}\label{eq:transmittable sources bound 3 simpl}
I(X_1 X_3;Y| X_2 S_2 V_2)&=H((V_1\+_2E_1)\+_4V_1\+_4N)-H(N)\\\label{eq:transmittable sources bound 4 simpl}
I(X_1 X_2 X_3;Y)&=H((E_1\+_2V_1\+_2V_2)\+_4(V_1\+_2 V_2)\+_4N)-H(N).
\end{align}
Since $V_1$ and $V_1\+_2V_2$ are both uniform over $\{0,1\}$, then the above two terms are equal. Let $\eta_2(\alpha)\deq 2-H((V_1\+_2E_1)\+_4V_1\+_4N)$. Note that $0\leq \eta_2(\alpha)\leq 2-H(N)$. Moreover, from Lemma \ref{lem:MAC in example is unstructured}, $\eta_2(\alpha)$ is strictly positive for any $\alpha \in (0, \frac{1}{2}]$. With this argument, the third and fourth conditions become 
\begin{align*}
h(\gamma)+h(\sigma)-h(\sigma * \gamma)\leq 2-H(N)-\eta_2(\alpha), \quad \text{and}~~
h(\gamma)+h(\sigma) \leq 2-H(N)-\eta_2(\alpha).
\end{align*}
Since the right-hand sides are equal and $h(\sigma * \gamma)\geq 0$, the third condition is trivial. 
 
As a result of the above argument, we obtain the following sufficient conditions:
\begin{subequations}\label{subeq:transmittable sources simp}
\begin{align}\label{subeq:transmittable sources simp 1}
h(\gamma)&\leq (1-\alpha)[2-H(N)]\\
h(\sigma) &\leq \eta_1(\alpha)\\
h(\gamma)+h(\sigma) &\leq 2-H(N)-\eta_2(\alpha)
\end{align}
\end{subequations}
For any $\gamma\leq \gamma^*$, inequality \eqref{subeq:transmittable sources simp 1} holds if $\alpha \leq 1-\frac{h_b(\gamma)}{h_b(\gamma^*)}$. Note that $\eta_1(\alpha)>0$ and $\eta_2(\alpha)> 0$ for all $\alpha \in (0, \frac{1}{2}]$, and $\eta_1(0)=\eta_2(0)= 0$. Further, they are continuous functions of $\alpha$ with $\lim_{\alpha \rightarrow 0} \eta_i(\alpha)=0, i=1,2$. Therefore, for any $\gamma<\gamma^*$, there exists $\alpha_0>0$ such that for any $\alpha\in (0,\alpha_0)$, inequality \eqref{subeq:transmittable sources simp 1} holds and $ h_b(\gamma^*)-h_b(\gamma)-\eta_2(\alpha)>0$. For any $\gamma\leq\gamma^*$, define
\begin{equation}
\sigma_0(\gamma)\deq h_b^{-1}\Big(\max_{0\leq \alpha \leq 1-h_b(\gamma)/h_b(\gamma^*) }  \min\big\{\eta_1(\alpha), ~h_b(\gamma^*)-\eta_2(\alpha)-h_b(\gamma)  \big\}  \Big).
\end{equation}
Note that the inequalities in \eqref{subeq:transmittable sources simp} are satisfied for $\gamma\leq \gamma^*$ and $\sigma=\sigma_0(\gamma)$. Hence, from the monotonicity of binary entropy function, these inequalities are also satisfied for $\sigma\leq \sigma_0(\gamma)$. This implies that any source with such parameters are transmissible. 
\end{proof}
The final step in our argument is as follows. Fix $\gamma \in (\gamma^*-\epsilon, \gamma^*)$, where $\epsilon$ is as in Lemma \ref{lem: neigh of gamma* is not achievable by CES}. From Lemma \ref{lem: neigh of gamma is achievable}, the source with such $\gamma$ and the parameter $\sigma = \sigma_0(\gamma)>0$ is transmissible; whereas from Lemma \ref{lem: neigh of gamma* is not achievable by CES} it is not transmissible using CES. Figure \ref{fig:CES_subopt} shows the set of parameters whose sources are transmissible. 

\begin{figure}[hbtp]
\centering
\includegraphics[scale=0.7]{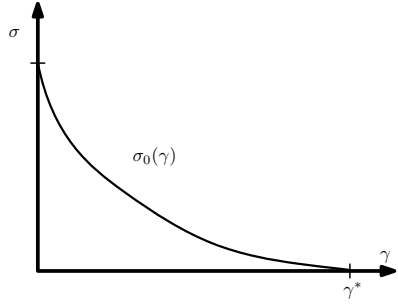}
\caption{The parameters $\sigma$ and $\gamma$ described in Lemma \ref{lem: neigh of gamma is achievable}.}
\label{fig:CES_subopt}
\end{figure}

%

\section{Proof of Theorem \ref{thm: MAC-FB achievable}}\label{sec: thm 1}
\subsection{Codebook Construction }

{ We build upon a class of codes called quasi linear codes (QLCs) \cite{QLC-ISIT16}. A QLC is defined as a subset of a linear code. By definition, any linear codebook can be viewed as the image of a linear transformation $\phi:\FF_q^k \mapsto \FF_q^n$, where $q$ is a prime number. In another words, the codewords of such a linear code are $\phi(\bfu^k), \bfu^k \in \FF_q^k$. In this representation,  a QGC over a finite field $\FF_q$ is defined as
  \begin{align} \label{eq: QLC codebook}
 \mathcal{C}\deq \{ \phi(\mathbf{u}): \mathbf{u}\in \mathcal{U}\}, 
\end{align}
where $\mathcal{U}$ is a given subset of $\FF_q^k$.  If $\mathcal{U}=\FF_q^k$, then $\mathcal{C}$ is a linear codebook.  

We begin the proof by generating a QGC for each user. Let $(W_1, W_2, W_3)$ be the random variables as in the statement of the theorem. For a fixed $\epsilon\in (0,1)$, consider the set of all $\epsilon$-typical sequences $\mathbf{w}_i^k$. Without loss of generality assume that the new message at the $i$th encoder is a sequence $\mathbf{w}_i^k$ which is selected randomly and uniformly from $A_\epsilon^{(k)}(W_i)$.  In this case $M_i=|A_\epsilon^{(k)}(W_i)|, i=1,2,3$.  }

We generate three codebooks for each user at each block $l\in [1, L]$. The codebook generations are described in the following:

{  \textbf{Codebook 1:}
For each block $l\in [1, L]$ generate $M_{0,[l]}$ sequences randomly and independently according to $P_U^n$. The parameter $M_{0,[l]}$ is to be defined later. Denote such sequences by $\bfU_{[l]}(m)$, where $m\in [1, M_{0,[l]}]$. 

 \textbf{Codebook 2:}
At each user $i=1,2,3$ and for any vector $\mathbf{w}_i^k \in \FF_2^k$, denote $$\mathbf{t}_i(\bfw_i^k)\deq \bfw_i^k \mathbf{G}\+\bfb_i^n, \quad i=1,2,3,$$
where $\mathbf{G}$ is a $k\times n$ matrix with elements chosen randomly and uniformly from $\FF_q$, and $\bfb_i^n $ is a vector selected randomly and uniformly from $\FF_q^n$. 

 \textbf{Codebook 3:}
For each user $i=1,2,3$ and given $\bfu^n \in \mathcal{U}^n$ and $\mathbf{t}^n ,\bfv^n \in \FF_q^n$ generate $M_i$ sequences randomly and independently according to the conditional distribution $\prod_{j=1}^n P(\cdot | u_j, t_j, v_j)$. Denote such sequences by $\bfx_i(\bfu^n, \bft^n, \bfv^n, m_i)$, where $m_i\in [1:M_i], i=1,2,3$. 

 }
 
\textbf{Initialization:}
Note that we are using the following notation: the subscript with bracket $[\cdot]$ denotes the index of a block, subscript without a bracket denotes the index of a user, and in line bracket $(\cdot)$ denotes the index of a codeword in a corresponding codebook. When it is clear from the context, we drop the index of the codewords.

For block $l=0$, set $M_{0,[0]}=1$. For block $l=1$, set $M_{0,[1]}=1, \mathbf{v}_{i,[1]}=\bf0$ for $i=1,2,3$. For block $l=2$, set  $M_{0,[2]}=1$. Let $\bf A \in \FF_q^{3\times 3}$, and by $a_{ij}$ denote the element in $i$th row and $j$th column. By $\mathbf{A}_i, i=1,2,3,$ denote the $i$th column of $\bf A$. At each block User $i$ intends to decode a linear combination of the messages with coefficients determined by $\mathbf{A}_i$. 

\subsection{Encoding and Decoding } 

{\textbf{Block $l=1$:}}
At block $l=1$, a new message $\mathbf{w}_{i, [1]} \in A_\epsilon^{(k)}(W_i)$, $i=1,2,3,$ is to observed by the $i$th user. Given the message, the $i$th encoder calculates $\bft_i(\mathbf{w}_{i,[1]})$. This sequence is denoted by $\bft_{i,[1]}$. Next, the encoder sends $\bfx_i(\bfu_{[1]},\bft_{i,[1]}, \bfv_{i,[1]}, \mathbf{w}_{i,[1]})$ over the channel. For shorthand, we denote such sequence by $\mathbf{x}_{i, [1]}$. The encoding and decoding processes in this block are shown in Table \ref{tab:MACFB User 1}.

{\textbf{Block $l=2$:}}

At the beginning of this block, each user receives $Y_{[1]}$ as feedback from the channel. User $i, i=1,2,3,$ wishes to decode the linear combination ${\bfw}_{\mathbf{A}_i, [1]}\deq a_{1i} \bfw_{1,[1]}\+a_{2i}\bfw_{2,[1]}\+a_{3i}\bfw_{3,[1]}$. 
Since, $\bfw_{i,[1]}$ is known at User $i$, then it finds  $\hat{\bfw}_{\mathbf{A}_i, [1]} \in A_\epsilon^{(k)}({W}_{\bf A_i}| \bfw_{i,[1]})$ such that
\begin{equation}\label{eq:MAC-FB decoding sum} 
 (\hat{\bfw}_{\mathbf{A}_i,[1]}\mathbf{G}\+\bfb_{\mathbf{A}_i}, Y_{[1]})\in A_{\epsilon}^{(n)}({T}_{\mathbf{ A}_i}, Y|\bfu_{[1]}, \bft_{1,[1]}, \bfx_{1,[1]}),
\end{equation}
where $W_{\mathbf{ A}_i}\deq a_{1i} W_1\+a_{2i}W_2\+a_{3i}W_3, \bfb_{\mathbf{A}_i}\deq a_{1i} \bfb_1\+a_{2i}\bfb_2\+ a_{3i}\bfb_3$ and $T_{\mathbf{ A}_i}\deq a_{1i} T_1\+a_{2i}T_2\+a_{3i}T_3$.
A decoding error $E_{i, [2]}, i=1,2,3,$ is declared if $\hat{\bfw}_{\mathbf{A}_i,[1]}$ is not found or is not unique. If it is unique, the encoder sets $\bfv_{i,[2]}=\hat{\bfw}_{\mathbf{A}_i, [1]}\mathbf{G}\+\bfb_{\mathbf{A}_i}$. Otherwise, $\bfv_{i,[2]}$ is generated at random from $\FF_q^n$.

Next, a new message $\mathbf{w}_{i, [2]}, i=1,2,3,$ is observed at the $i$th encoder. Similar to the encoding process at the first block, the $i$th encoder calculates $\bft_{i,[2]}$ and sends $\bfx_i(\bfu_{[2,1]},\bft_{i,[2]}, \bfv_{i,[2]}, \mathbf{w}_{i,[2]})$. For shorthand, such sequence is denoted by $\mathbf{x}_{i, [2]}$.  The encoding and decoding processes in this block are shown in Table \ref{tab:MACFB User 1}.

{\textbf{Block $l>2$:}}
Each user performs two decoding and three encoding processes in this block.
It is assumed that each encoder knows the common information  given by $\bfu_{[l-2]}$ and $\bfu_{[l-1]}$. For $l=3$, this is clear because $M_{0,[1]]}=M_{0,[2]}=1$. We will explain 
how this knowledge is acquired, and 
how $\bfu_{[l]}$ is generated after describing the decoding 
process.

 The first decoding process is the same as the decoding process in block $l=2$. At the beginning of the block $l>2$, User $i$ observes $Y_{[l-1]}$ as feedback from the channel and wishes to decode the linear combination ${\bfw}_{\mathbf{A}_i, [l-1]}\deq a_{1i} \bfw_{1,[l-1]}\+a_{2i}\bfw_{2,[l-1]}\+a_{3i}\bfw_{3,[l-1]}$. This decoding process is the same as the one in block $l=2$; it is successful, if the sequences $\hat{\bfw}_{\mathbf{A}_i,[l-1]}$ is unique.  Then, the codeword $\bfv_{i,[l]}$ is generated at User $i$, where $i=1,2,3$. If the decoding process at User $i, i=1,2,3,$ is not successful, an error event $E_{i, [l]}$ is declared and a codeword $\bfv_{i,[l]}$ is generated at random.

{ Next, we explain the second decoding process. Given $(\bfY_{[l-2]}, \bfY_{[l-1]})$, User $i$ decodes the messages of the other two encoders from block $l-2$. For that, User $1$ finds unique $\hat{\bfw}_{2,[l-2]}\in A_\epsilon^{(k)}(W_2)$ and $\hat{\bfw}_{3,[l-2]}\in A_\epsilon^{(k)}(W_3)$ such that 
\begin{subequations}\label{eq: Type 2 decoding error}
 \begin{align}\label{eq: Type 2 decoding error eq 1}
&a_{1,1}\bfw_{1,[l-2]} \+ a_{2,1} \hat{\bfw}_{2,[l-2]}\+ a_{3,1}\hat{\bfw}_{3,[l-2]}=\hat{\bfw}_{\mathbf{A}_1, [l-1]}, \quad \text{and}\\\nonumber
\Big(& \hat{\bft}_{2,[l-2]}, \hat{\bfx}_{2,[l-2]}, \hat{\bft}_{3,[l-2]}, \hat{\bfx}_{3,[l-2]},\hat{\bfv}_{2,[l-1]},\hat{\bfv}_{3,[l-1]},\bfY_{[l-2]}, \bfY_{[l-1]} \Big)\\\label{eq: Type 2 decoding error eq 2}
& \hspace{80pt}  \in A_\epsilon^{(n)}\Big(\tilde{T}_2\tilde{X}_2\tilde{T}_3\tilde{X}_3V_2 V_3\tilde{Y}Y |  \bfs_{1,[l-2]},  \bfs_{1,[l-1]}, \bfv_{2,[l-2]}, \bfv_{3,[l-2]}, \bfu_{[l-1]}, \bfu_{[l-2]}\Big)
%
%
\end{align} 
\end{subequations}
where $\bfu_{[l-1]}, \bfu_{[l-2]}, \bfv_{i,[l-2]}$ are known at the encoder from the previous blocks and
\begin{align*}
\hat{\bft}_{i,[l-2]}&\deq \bft_i(\hat{\bfw}_{i,[l-2]}),\\
\hat{\bfx}_{i,[l-2]}&\deq \bfx_i\big(\bfu_{[l-2]}, \hat{\bft}_{i,[l-2]}, \bfv_{i,[l-2]},\hat{\bfw}_{i,[l-2]}\big),\\
\hat{\bfx}_{i,[l-2]}&\deq \bfx_i\big(\bfu, \hat{\bft}_{i,[l-2]}, \bfv_{i,[l-2]},\hat{\bfw}_{i,[l-2]}\big),\\
\hat{\bfv}_{2,[l-1]}&\deq (a_{1,2}\bfw_{1,[l-2]} \+ a_{2,2} \hat{\bfw}_{2,[l-2]}\+ a_{3,2}\hat{\bfw}_{3,[l-2]})\mathbf{G}\+\bfb_{\mathbf{A}_2},\\
\hat{\bfv}_{3,[l-1]}&\deq (a_{1,3}\bfw_{1,[l-2]} \+ a_{2,3} \hat{\bfw}_{2,[l-2]}\+ a_{3,3}\hat{\bfw}_{3,[l-2]})\mathbf{G}\+\bfb_{\mathbf{A}_3}.
\end{align*}
If the messages are not unique, an error event will be declared. { This decoding process is repeated for User 2 and 3.  With these decoding processes each user obtains an estimate of the messages of the other two users.  By $\tilde{E}_{i,[l]}$ denote the error event in the second phase of the decoding process at User $i$ and block $l$.}

{ Next, the transmitters and the receiver generate a common list of highly likely messages for block $l-2$. In what follows, we define this list.  For any triplet of the messages $(\tilde{\bfw}_1,\tilde{\bfw}_2,\tilde{\bfw}_3)$ let 
\begin{align*}
\tilde{\bfx}_{i,[l-2]}(\tilde{\bfw}_i)&\deq \bfx_i\big(\bfu_{[l-2]}, \bft_i(\tilde{\bfw}_{i}), \bfv_{i,[l-2]},\tilde{\bfw}_{i}\big)
\end{align*}
where $\bfu_{[l-2]}$ and $\bfv_{i,[l-2]}, i=1,2,3,$ are known from previous block. For shorthand denote $$\underline{\tilde{\bfx}}_{[l-2]}(\underline{\tilde{\bfw}})\deq \big(\tilde{\bfx}_{i,[l-2]}(\tilde{\bfw}_i)\big)_{i=1,2,3}, \qquad \underline{\tilde{\bft}}(\underline{\tilde{\bfw}})\deq \big(\bft_i(\tilde{\bfw}_{i}) \big)_{i=1,2,3}.$$
Next, given the channel output $Y_{[l-2]}$, define the list of highly likely messages corresponding to block $l-2$ as 
\begin{align} \label{eq:L list}
\mathcal{L}[l-2] \deq \Big\{\underline{\tilde{\bfw}} &\in A_\epsilon^{(n)}(W_1,W_2,W_3): \big(Y_{[l-2]}, \bfu_{[l-2]}, \underline{\tilde{\bfx}}_{[l-2]}(\underline{\tilde{\bfw}}),  \underline{\tilde{\bft}}(\underline{\tilde{\bfw}})\big) \in A_\epsilon^{(n)}(\tilde{Y},\tilde{U}, \underline{\tilde{X}}, \underline{\tilde{T}} )\Big\}
\end{align}  }
where $\underline{\tilde{\bfw}}\deq (\tilde{\bfw}_1,\tilde{\bfw}_2,\tilde{\bfw}_3), ~ \underline{\tilde{X}}\deq (\tilde{X}_1, \tilde{X}_2, \tilde{X}_3)$ and $\underline{\tilde{T}}\deq (\tilde{T}_1, \tilde{T}_2, \tilde{T}_3)$. Note that the set $\mathcal{L}[l-2]$ represents the uncertainty of the receiver about the transmitted messages at block $l-2$. This list can be calculated at the transmitters as well as the receiver. Set $M_{0, [l]}=|\mathcal{L}[l-2]|$ as the size of codebook 1. Index all members of $\mathcal{L}[l-2]$ by $m\in [1, M_{0, [l]}]$.

Suppose the decoding processes in the transmitters are successful, which means the messages are estimated correctly.  Suppose $\hat{\bfw}_{2,[l-2]},\hat{\bfw}_{3,[l-2]}$ are the estimated messages at User 1.  If $(\bfw_{1,[l-2]},\hat{\bfw}_{2,[l-2]},\hat{\bfw}_{3,[l-2]}) \in \mathcal{L}[l-2]$, then the first encoder finds its index (say $m_1$) in $\mathcal{L}[l-2]$. Similarly, User 2 and 3 find the index of their estimated messages (say $m_2$ and $m_3$).  Since the decoding processes are assumed to be successful, these indices are equal, i.e., $m_1=m_2=m_3=m$. Therefore, the transmitters can calculate the corresponding codeword in codebook 1, i.e., $\bfu_{[l]}(m)$. Note that the receiver is not able to find $m$. This is because each transmitter knows its own message and has less uncertainty comparing to the receiver.  The objective of Codebook 1 is to resolve the uncertainty at the decoder. 

The next step is the encoding process for block $l$ which is similar to the previous blocks. Given a new message $\mathbf{w}_{i, [2]}, i=1,2,3,$ at User $i$, the sequence $\bft_{i,[l]}$ is calculated and  the codeword $\bfx_i(\bfu_{[l]},\bft_{i,[l]}, \bfv_{i,[l]}, \mathbf{w}_{i,[l]})$ is sent to the channel. For shorthand, the transmitted codeword is denoted by $\mathbf{x}_{i, [l]}$.  The encoding and decoding processes in this block are shown in Table \ref{tab:MACFB User 1} and \ref{tab:MACFB User i}. }

\begin{table}
\centering
\caption{The decoding and encoding processes for User 1 in blocks $l=1,2,3$.}
\label{tab:MACFB User 1}
\begin{tabular}{|c|c|c|c|}
\cline{2-4}
    \multicolumn{1}{c|}{} & $l=1$ & $l=2$ & $l=3$  \\ [1ex]
\hline 
Decoding 1 & $\bfv_{1, [1]}=\bf0$ & \begin{tabular}{c} $\hat{\bfw}_{\mathbf{A}_1,[1]},~ \bfv_{1, [2]}={\bft}_{\mathbf{A}_1,[1]}$\end{tabular} &${\bfw}_{\mathbf{A}_1,[2]}$, $\bfv_{1, [3]}=\bft_{\mathbf{A}_1, [2]}$ \\  [1ex]
\hline \rule[-1ex]{0pt}{2.5ex}
Decoding 2 & --- & --- & $(\hat{\bfw}_{2,[1]},\hat{\bfw}_{3,[1]})$ \\ [1ex]
\hline \hline
Encoding 1 & $M_{0,[1]}=1, \bfu_{[1]}$ & $M_{0,[2]}=1, \bfu_{[2]}$ & $\mathcal{L}_{[1]},~ M_{0,[3]}=|\mathcal{L}_{[1]}|,  \bfu_{[3]}$  \\ [1ex]
\hline 
Encoding 2 & $\bfw_{1, [1]}, \bft_{1,[1]}(\bfw_{1, [1]})$ & $\bfw_{1, [2]}, \bft_{1,[2]}(\bfw_{1, [2]})$ & $\bfw_{1, [3]}, \bft_{1,[3]}(\bfw_{1, [3]})$ \\ [1ex]
\hline 
Encoding 3 & $\bfx_1(\bfu_{[1]}, \bft_{1,[1]}, \bfv_{1,[1]}, \bfw_{1, [1]}) $ & $\bfx_1(\bfu_{[2]}, \bft_{1,[2]}, \bfv_{1,[2]}, \bfw_{1, [2]}) $ & $\bfx_1(\bfu_{[3]}, \bft_{1,[3]}, \bfv_{1,[3]}, \bfw_{1, [3]}) $  \\ [1ex]
\hline 
\end{tabular} 

\end{table}

\begin{table}[!htbp]
\centering
\caption{The decoding and encoding processes for User i in block $l$.}
\label{tab:MACFB User i}
\begin{tabular}{|c|c|c|c|}
\cline{2-2}
    \multicolumn{1}{c|}{} & block: $l$ \\ 
\hline 
Decoding 1 &  \begin{tabular}{c} $\hat{\bfw}_{\mathbf{A}_i,[l-1]}= a_{1i} \bfw_{1,[l-1]}\+a_{2i}\bfw_{2,[l-1]}\+a_{3i}\bfw_{3,[l-1]}$\\ $\bfv_{i, [l]}={\bft}_{\mathbf{A}_i,[l-1]}$\end{tabular} \\ [1ex]
\hline
Decoding 2 & $(\hat{\bfw}_{j,[l-2]},\hat{\bfw}_{k,[l-2]})$ \\[1ex] 
\hline \hline
Encoding 1 & $\mathcal{L}_{[l-2]}, \bfu_{[l]}$  \\ [1ex]
\hline 
Encoding 2 & $\bfw_{i, [l]}, \bft_{i,[l]}(\bfw_{i, [l]})$ \\ [1ex]
\hline 
Encoding 3 & $\bfx_i(\bfu_{[l]}, \bft_{i,[l]}, \bfv_{i,[l]}, \bfw_{i, [l]}) $ \\ [1ex]
\hline 
\end{tabular} 

\end{table}

\paragraph*{\textbf{Decoding at block $l$}}
The decoder knows the list of highly likely messages. This list is $\mathcal{L}[l-2]$ as defined in \eqref{eq:L list}. Given $Y_{[l]}$ the decoder wishes to decode $U_{[l]}$ using which it can find the transmitted messages at block $l-2$. This decoding process is performed by finding an index $m\in [1: M_{0,[l]}]$ such that $$(U_{[l,m]}, Y_{[l]}) \in A_\epsilon^{(n)}(U, Y).$$
If the index is not found or is not unique, then an error event $E_{d, [l]}$ is declared. 

\subsection{Error Analysis}
There are three types of decoding errors:

{ \qquad  1. Error in decoding the linear combination of the messages, i.e., ${E}_{i,[l]}, i=1,2,3, l \geq 2$.

\qquad 2. Error in the decoding of the messages of the other encoders, i.e., $\tilde{E}_{i,[l]}, i=1,2,3, l \geq 3$.

\qquad  3. Error at the decoder, i.e. ${E}_{d,[l]}, l \geq 3$.  }

{ The total error probability is the probability of the union of above error events: 
\begin{align}\nonumber
P_e&=\PP\left\{\bigcup_{l\geq 2} \left({E}_{d,[l]} \medcup \left[ \medcup_{i=1}^3 {E}_{i,[l]} \medcup \tilde{E}_{i,[l]}\right] \right)\right\}\\\nonumber
&\leq B~ \PP\left\{{E}_{d,[3]} \medcup \left[ \medcup_{i=1}^3 {E}_{i,[3]} \medcup \tilde{E}_{i,[3]}\right]\right\}\\\nonumber
&\leq B~ \PP\left\{\medcup_{i=1}^3 {E}_{i,[3]} \medcup \tilde{E}_{i,[3]}\right\} + B~ \PP\left\{{E}_{d,[3]}~ \Big|~ \medcap_{i=1}^3 {E}^c_{i,[3]} \medcap \tilde{E}^c_{i,[3]}\right\}\\\label{eq:MACFB Pe bound}
&\leq B~ \sum_{i=1}^3 \left[  \PP\{{E}_{i,[3]}\} +\PP\{ \tilde{E}_{i,[3]} ~\big|~  {E}^c_{i,[3]}\}\right] + B~ \PP\left\{{E}_{d,[3]} ~ \Big|~  \medcap_{i=1}^3 {E}^c_{i,[3]} \medcap \tilde{E}^c_{i,[3]}\right\},
\end{align}
where $B$ is the number of blocks. The first inequality holds due to the union bound on $l$ and the fact that $l$ does not change the probability of the error events. The second and third inequality hold because $P(A\medcup B) \leq P(A)+P(B| A^c)$ and the union bound on $i$.
Using standard arguments for each type of the errors we get the following bounds:

The probability of the first type of the errors ($\PP\{{E}_{i,[3]}\}$) can be made arbitrary small for sufficiently large $n$, if for any distinct $i,j,k \in \{1,2,3\}$ the following bound holds: }
\begin{align}\label{eq: bound 1}
\frac{k}{n}H(W_{\mathbf{A}_i}|W_i) \leq I(T_{\mathbf{A}_i}; Y|U T_i V_i X_i)-\delta_1(\epsilon).
\end{align}
The argument follows by standard error analysis for decoding $\bfw_{\mathbf{A}_i}$ at User $i$. At User $i$, with probability  sufficiently close to 1, $\bfw_{\mathbf{A}_i}$ satisfies \eqref{eq:MAC-FB decoding sum}. Hence, to analyze ${E}_{i,[3]}$, it suffices to find the probability that a codewrod $\hat{\bfw}_{\mathbf{A}_i} \neq \bfw_{\mathbf{A}_i}$ satisfies \eqref{eq:MAC-FB decoding sum}. Note that $\bfw_i$ is known at User $i$. Hence, there are approximately $2^{kH(W_{\mathbf{A}_i}|W_i)}$ $\epsilon$-typical sequences $\hat{\bfw}_{\mathbf{A}_i}$. From standard arguments, one can show that the probability that each of such sequences satisfies \eqref{eq:MAC-FB decoding sum} is approximately equals to $2^{-nI}$, where $I$ is the mutual information on the right-hand side of \eqref{eq: bound 1}. Therefore, the error probability $\PP\{{E}_{1,[3]}\}$ approaches zero, if \eqref{eq: bound 1} is satisfied. 

The probability of the second type of the errors ($\PP\{ \tilde{E}_{i,[3]}  |   {E}^c_{i,[3]}\} $) approaches zero for sufficiently large $n$, if
{ \begin{align}\label{eq: bound 2}
\frac{k}{n}H(W_j, W_k|W_i, W_{\mathbf{A}_i})\leq I(\tilde{T}_j \tilde{X}_j \tilde{T}_k \tilde{X}_k; Y \tilde{Y}| \tilde{U} \tilde{S}_i {U} {S}_i  \tilde{V}_j \tilde{V}_k )-\delta_2(\epsilon).
\end{align}
For this type of error it is assumed that the linear combination $\bfw_{\mathbf{A}_i}$ is decoded correctly. Hence, one needs to find the probability that \eqref{eq: Type 2 decoding error} is satisfied for a pair $(\hat{\bfw}_j, \hat{\bfw}_k)\neq ({\bfw}_j, {\bfw}_k)$. There are approximately $2^{kH(W_j, W_k|W_i, W_{\mathbf{A}_i})}$ such jointly typical pairs satisfying \eqref{eq: Type 2 decoding error eq 1}. The probability that any of such pairs satisfies \eqref{eq: Type 2 decoding error eq 2} is sufficiently small for large enough $n$ if the following inequality holds
\begin{equation*}
\frac{k}{n}H(W_j, W_k|W_i, W_{\mathbf{A}_i})\leq I(\tilde{T}_j \tilde{X}_j \tilde{T}_k \tilde{X}_k V_j  V_k; Y \tilde{Y}| \tilde{U} \tilde{S}_i {U} {S}_i  \tilde{V}_j \tilde{V}_k )-\delta_3(\epsilon)
\end{equation*}
The mutual information above equals to the one in \eqref{eq: bound 2}. This is due to the fact that $V_i=\tilde{T}_{\mathbf{A}_i}$, as stated below the equation in \eqref{eq:P P tilde joint dist}.  
 }

The third type of error ($\PP\{{E}_{d,[3]} | ~\medcap_{i=1}^3 {E}^c_{i,[3]} \medcap \tilde{E}^c_{i,[3]}\}$) approaches zero, if { $|\mathcal{L}[l]| < 2^{nI(U;Y)}$}. It can be shown that for sufficiently large $n$,
\begin{align*}
\PP\left\{|\mathcal{L}[l]| < 2^{n\max_{\mathcal{B}\subseteq \{1,2,3\}} F_{\mathcal{B}} +o(\epsilon)} \right\} > 1-\epsilon, 
\end{align*}
where 
\begin{equation}\label{eq:F_A def}
F_{\mathcal{B}} \deq \frac{k}{n} H(W_{\mathcal{B}})-I(X_\mathcal{B}; Y| U S_{\mathcal{B}^c} \tilde{V}_1, \tilde{V}_2, \tilde{V}_3), \quad \forall \mathcal{B}\subseteq \{1,2,3\}. 
\end{equation}
Therefore, the probability of third type of the errors approaches zero with rate $2^{-n\delta}$ for $\delta\in (0,1)$ and sufficiently large $n$, if the following bounds hold for any subset  $\mathcal{B}\subseteq \{1,2,3\}$:
{ \begin{align*}
F_{\mathcal{B}}\leq I(U;Y)-\delta-o(\epsilon),
\end{align*}}
Using the definition of $F_{\mathcal{B}}$ in \eqref{eq:F_A def}, the above bounds are equivalent to the following:
{ \begin{align} \label{eq: bound 3}
\frac{k}{n} H(W_{\mathcal{B}}) \leq I(X_\mathcal{B}; Y| U S_{\mathcal{B}^c} \tilde{V}_1, \tilde{V}_2, \tilde{V}_3)+ I(U;Y)-\delta-o(\epsilon), \quad \forall \mathcal{B}\subseteq \{1,2,3\}
\end{align}}
{ Consequently, if the bounds in \eqref{eq: bound 1}, \eqref{eq: bound 2}, and \eqref{eq: bound 3} are satisfied for a fixed $\delta>0$, then, from the inequality in \eqref{eq:MACFB Pe bound}, we obtain 
\begin{align*}
P_e\leq 7 B 2^{-n\delta}
\end{align*}
Hence, if $B$ grows sub-exponentially as a function of $n$, then $P_e\rightarrow 0$ as $n \rightarrow \infty$.

Note that the effective rate of our coding scheme is $R_i \deq \frac{1}{n}\log_2 M_i= \frac{k}{n}H(W_i)$ for $i=1,2,3$. Therefore, from the bounds in \eqref{eq: bound 1}, \eqref{eq: bound 2}, and \eqref{eq: bound 3}, a rate triplet $(R_1,R_2,R_3)$ is achievable if there exist $\alpha \in (0,1)$ and random variables ${W}_r, {T_r}, {V_r}, {X_r}, {\tilde{T_r}}, {\tilde{V_r}}, {\tilde{X_r}}, r=1,2,3 $, distributed as described in Theorem 2,  such that 
\begin{align*}
\alpha H(W_i)&= R_i,\\
\alpha H(W_{\mathbf{A}_i}|W_i) &\leq I(T_{\mathbf{A}_i}; Y|U T_i V_i X_i),\\
\alpha H(W_j,W_k|W_{\mathbf{A}_i}, W_i)& \leq I(\tilde{T}_j \tilde{X}_j \tilde{T}_k \tilde{X}_k; Y \tilde{Y}| \tilde{U} \tilde{S}_i {U} {S}_i  \tilde{V}_j \tilde{V}_k ),\\
\alpha H(W_{\mathcal{B}}) &\leq I(X_\mathcal{B}; Y| U S_{\mathcal{B}^c} \tilde{V}_1, \tilde{V}_2, \tilde{V}_3)+ I(U;Y)
\end{align*}}

\section{Proof of Theorem \ref{thm:MAC-FB structured}}\label{appx:MAC-FB structured proof} 
We begin the proof by the following lemma.
\begin{lem}\label{lem: example achievable rate using linear codes}
For the channel given in Example \ref{ex: example}, the rate triple $(1-h(\delta), 1-h(\delta), 1-h(\delta))$ is achievable.
\end{lem} 
\begin{proof}
The proof is given in Appendix \ref{sec: proof of lemma 1 of the example}.
\end{proof}
\begin{remark}\label{rem: optimality}
{{The triple $(1-h(\delta), 1-h(\delta), 1-h(\delta))$ is a corner point in the capacity region of the channel in Example \ref{ex: example}. This implies the optimality of the above coding strategy in terms of achievable rates. }}
\end{remark}

The above coding strategy is different from known schemes in two ways: 1) Identical linear codes are used to encode the messages, 2) The third user uses feedback to decode only the binary sum others' messages.  

One implication of Remark \ref{rem: optimality} is that the proposed coding scheme achieves optimality.  We show a stronger result in this Subsection. We prove that every coding scheme that achieves  $(1-h(\delta), 1-h(\delta), 1-h(\delta))$, should carry certain algebraic structures such as closeness under the binary addition. 

Suppose there exists a $(N, M_1, M_2, M_3)$ transmission system with rates close to $R_i=1-h(\delta)$, and average probability of error close to $0$, in particular
\begin{align*}
\bar{P}<\epsilon, \quad \frac{1}{n}\log_2 M_i \geq 1-h(\delta)-\epsilon, \quad i=1,2,3,
\end{align*} 
where $\epsilon>0$ is sufficiently small. Since there is no feedback at the first and second encoder, the transmission system predetermines a codebook for user 1 and 2. Note that there are two outputs for encoder 1 and 2. Suppose $\mathcal{C}_{12}$ and $\mathcal{C}_{22}$ are the codebooks assigned to the second output of encoder 1 and encoder 2, respectively. 

Let $\mathbf{X}^N_{i2}$ be the second output of encoder $i$, where $i=1,2,3$.  Let $X_{i2, l}$ denote the $l$th component of $X^N_{i2}$, where  $1\leq l \leq N, ~ i=1,2,3$.  The following lemmas hold for this transmission system. 

\begin{lem}\label{lem: x_2+x_1 needs to be decoded}
 For any fixed $c>0$, define
\begin{align*}
\mathcal{I}_c^N:=\{ l\in [1:N]: P(X_{32, l} \neq X_{12, l}\oplus X_{22, l}) \geq c \}.
\end{align*}
 Then, the inequality $\frac{|\mathcal{I}_c^N|}{N}\leq   \frac{\eta(\epsilon)}{2c(1-h(\delta))}$ holds, where $\eta(\epsilon)$ is a function such that,  $\eta(\epsilon) \rightarrow 0$, as $\epsilon\rightarrow 0$. 
\end{lem}
\begin{proof}
The proof is given in Appendix \ref{sec: proof of lemma 2 of the example}.
\end{proof}
The Lemma implies that in order to achieve $(1-h(\delta), 1-h(\delta), 1-h(\delta))$, the third user needs to decode $ X_{12, l}\oplus X_{22, l}$ for ``almost all" $l\in [1:N]$. This requirement is necessary to insure that the channel given in Figure \ref{fig: Exp. second chann} is in the first state. 

In the next step, we use the results of Lemma \ref{lem: x_2+x_1 needs to be decoded}, and  drive two necessary conditions for decoding $X_{12}\oplus X_{22}$.

\begin{lem}\label{lem: structure in the code}
The following holds 
\begin{align*}
\frac{1}{N}\big| ~\log ||\mathcal{C}_{12}\oplus \mathcal{C}_{22}||- \log ||\mathcal{C}_{12}||~\big| \leq \lambda_1(\epsilon),\\
\frac{1}{N}\big| ~\log ||\mathcal{C}_{12}\oplus \mathcal{C}_{22}||- \log ||\mathcal{C}_{22}||~\big| \leq \lambda_2(\epsilon),
\end{align*}
where $\lambda_j(\epsilon) \rightarrow 0$, as $\epsilon\rightarrow 0, j=1,2$.
\end{lem}
\begin{proof}
The proof is given in Appendix \ref{seq: lem 3}.
\end{proof}
As a result of this lemma,  $\log ||\mathcal{C}_{12}\oplus \mathcal{C}_{22}||$ needs to be close to $\log ||\mathcal{C}_{12}||$ and $\log ||\mathcal{C}_{22}||$. This implies that $\mathcal{C}_{12}$ and $\mathcal{C}_{22}$ possesses an algebraic structure, and are \textit{almost} close under the binary addition. Not that for the case of unstructured random codes $||\mathcal{C}_{12}\oplus \mathcal{C}_{22}||\approx ||\mathcal{C}_{12}||\times ||\mathcal{C}_{22}||$. Hence, unstructured random coding schemes are suboptimal in this example.

\begin{remark}
The three-user extension of CL scheme is suboptimal. Because, the conditions in Lemma \ref{lem: structure in the code} are not satisfied. 
\end{remark}


\section{Proof of Lemma \ref{lem: example achievable rate using linear codes} to \ref{lem: structure in the code} }
\subsection{Proof of Lemma \ref{lem: example achievable rate using linear codes}}\label{sec: proof of lemma 1 of the example}
\begin{proof}[Outline of the proof]
We start by proposing a coding scheme. There are $L$ blocks of transmissions in this scheme, with new messages available at each user at the beginning of each block. The scheme sends the messages with $n$ uses of the channel. Let $\mathbf{W}^k_{i,[l]}$ denotes the message of the $i$th transmitter at the $l$th block, where $i=1,2,3$, and $1\leq l \leq L$. Let  $\mathbf{W}^k_{i,[l]}$ take values randomly and uniformly from $\FF_2^k$. In this case, the transmission rate of each user is $R_i=\frac{k}{n}, i=1,2,3$. The first and the second outputs of the $i$th encoder in block $l$ is denoted by $\mathbf{X}^n_{i1,[l]}$ and $\mathbf{X}^n_{i2,[l]}$, respectively.

\textbf{Codebook Construction:}
Select a $k\times n$ matrix $\mathbf{G}$ randomly and uniformly from $\FF_2^{k\times n}$. This matrix is used as the generator matrix of a linear code. Each encoder is given the matrix $\mathbf{G}$. Therefore, the encoders use an identical linear code generated by $\mathbf{G}$.

\textbf{Encoder 1 and 2:}
For the first block set $\mathbf{X}^n_{i2,[1]}=0$, for  $i=1,2,3$. For the block $l$, encoder 1 sends $\mathbf{X}^n_{11,[l]}= \mathbf{W}^k_{1,[l]} \mathbf{G}$ through its first output. For the second output, encoder 1 sends  $\mathbf{X}^n_{11,[l-1]}$ from block $l-1$, that is $\mathbf{X}^n_{12,[l]}=\mathbf{X}^n_{11,[l-1]}$. Similarly, the outputs of the second encoder are  $\mathbf{X}^n_{21,[l]}= \mathbf{W}^k_{2,[l]} \mathbf{G}$, and   $\mathbf{X}^n_{22,[l]}=\mathbf{X}^n_{21,[l-1]}$. 
  
\textbf{Encoder 3:}
The third encoder sends  $\mathbf{X}^n_{31,[l]}= \mathbf{W}^k_{3,[l]} \mathbf{G}$ though its first output. This encoder receives  the feedback from the block $l-1$ of the channel. This encoder wishes to decode $\mathbf{W}^k_{1,[l-1]}\oplus\mathbf{W}^k_{2,[l-1]}$ using $\mathbf{Y}^n_{1,[l-1]}$. For this purpose,  it subtracts  $\mathbf{X}^n_{31,[l-1]}$ from $\mathbf{Y}^n_{1,[l-1]}$. Denote the resulting vector by $\mathbf{Z}^n$. Then, it finds a unique vector $\mathbf{\tilde{w}}^k \in \FF_2^k $ such that $(\mathbf{\tilde{w}}^k\mathbf{G}, \mathbf{Z}^n)$ is $\epsilon$-typical with respect to $P_{XZ}$, where $X$ is uniform over $\FF_2$ , and $Z=X\oplus \tilde{N}_\delta$. If the decoding process is successful, the third encoder sends $\mathbf{X}^n_{32,[l]}= \mathbf{\tilde{w}}^k_{[l-1]} \mathbf{G}$. Otherwise, an event $E_{1, [l]}$ is declared.

\textbf{Decoder:}
The decoder receives the outputs of the channel from the $l$th block, that is $\mathbf{Y}^n_{1,[l]}$ and $\mathbf{Y}^n_{2,[l]}$. The decoding is performed in three steps. First, the decoder uses $\mathbf{Y}^n_{2,[l]}$ to decode $ \mathbf{W}^k_{1,[l-1]}$, and $ \mathbf{W}^k_{2,[l-1]}$. In particular, it finds unique $\mathbf{\tilde{w}}_1^k, \mathbf{\tilde{w}}_2^k \in \FF_2^k$ such that $(\mathbf{\tilde{w}}_1^k\mathbf{G}, \mathbf{\tilde{w}}_2^k\mathbf{G}, \mathbf{Y}^n_{2,[l]})$ are jointly $\epsilon$-typical with respect to $P_{X_{12}X_{22}Y_2}$. Otherwise, an error event $E_{2,[l]}$ will be declared. 

Suppose the first part of the decoding process is successful. At the second step,  the decoder calculates $\mathbf{X}^n_{11,[l-1]}$, and $\mathbf{X}^n_{21,[l-1]}$. This is possible, because  $\mathbf{X}^n_{11,[l-1]}$, and $\mathbf{X}^n_{21,[l-1]}$ are functions of the messages.  The decoder, then, subtracts $\mathbf{X}^n_{11,[l-1]}\oplus \mathbf{X}^n_{21,[l-1]}$ from $Y_{1,[l-1]}$. The resulting vector is 
\begin{align*}
\tilde{\mathbf{Y}}^n=\mathbf{X}^n_{31,[l-1]}\oplus \tilde{N}^n_\delta.
\end{align*}
In this situation, the channel from $X_{31}$ to $\tilde{Y}$ is a binary additive channel with $\delta$ as the bias of the noise. At the third step, the decoder uses $\tilde{\mathbf{Y}}^n$ to decode the message of the third user, i.e., $\mathbf{W}^k_{3,[l-1]}$. In particular, the decoder finds unique $\mathbf{\tilde{w}}_3^k \in \FF_2^k$ such that $(\mathbf{\tilde{w}}_3^k\mathbf{G},  \mathbf{\tilde{Y}}^n)$ are jointly $\epsilon$-typical with respect to $P_{X_{31}\tilde{Y}}$. Otherwise, an error event $E_{3,[l]}$ is declared.

\textbf{Error Analysis:} 
We can show that this problem is equivalent to a point-to-point channel coding problem, where the channel is described by $Z=X\oplus \tilde{N}_\delta$. The average probability of error approaches zero, if $\frac{k}{n}\leq 1-h_b(\delta)$. 

Suppose there is no error in the decoding process of the third user. That is $E_{1,[l]}^c$ occurs. Therefore, $\mathbf{X}^n_{32,[l]}=\mathbf{X}^n_{22,[l]}\oplus \mathbf{X}^n_{12,[l]}$ with probability one. As a result, the channel in Fig. \ref{fig: Exp. second chann} is in the first state. This implies that the corresponding channel consists of two parallel binary additive channel with independent noises and bias $\delta$. Similar to the argument for $E_1$, it can be shown that $P(E_{2,[l]}| E_{1,[l]})\rightarrow 0$, if $\frac{k}{n}\leq 1-h_b(\delta)$. Lastly, we can show that conditioned on $E_{1,[l]}^c$ and $E_{2,[l]}^c$,  the probability of $E_{3,[l]}$ approaches zero, if $\frac{k}{n}\leq 1-h_b(\delta)$. 

As a result of the above argument, the average probability of error approaches $0$, if $\frac{k}{n}\leq 1-h_b(\delta)$. This implies that the rates $R_i=1-h_b(\delta), i=1,2,3$ are achievable, and the proof is completed.
\end{proof}

\subsection{Proof of Lemma \ref{lem: x_2+x_1 needs to be decoded}}\label{sec: proof of lemma 2 of the example}
\begin{proof}
Let $R_i$ be the rate of the $i$th encoder. We have $R_i\geq 1-h_b(\delta)-\epsilon$. We apply the generalized Fano's inequality (Lemma 4.3 in \cite{Kramer-thesis}) for decoding of the messages. More precisely, as $\bar{P}\leq \epsilon$, we have $$\frac{1}{M_1M_2M_3}H(\Theta_1, \Theta_2, \Theta_3| \mathbf{Y}^N)\leq h(\bar{P}) \leq h(\epsilon)$$

By the definition of the rate we have 
\begin{align}\nonumber
R_1+R_2+R_3&=\frac{1}{N}H(\Theta_1, \Theta_2, \Theta_3)\\\nonumber
& \leq \frac{1}{N}I(\Theta_1, \Theta_2, \Theta_3; \mathbf{Y}^n)+o(\epsilon)\\\nonumber
&\stackrel{(a)}{\leq } \frac{1}{N}I(\mathbf{X}^n_1, \mathbf{X}^n_2, \mathbf{X}^n_3; \mathbf{Y}^N)+o(\epsilon)\\\label{eq: last bound}
&\stackrel{(b)}{\leq } 3-\frac{1}{N}H(\mathbf{Y}^n|\mathbf{X}^n)+o(\epsilon),
\end{align}
where  $(a)$ is because of (\ref{eq: chann probabilities}), and for $(b)$ we use the fact that $Y$ is a vector of three binary random variables, which implies$\frac{1}{N}H(Y^N)\leq 3$.
As the channel is memoryless, and since (\ref{eq: chann probabilities}) holds, we have 
\begin{align*}
\frac{1}{N}H(\mathbf{Y}^n|\mathbf{X}^n)=\frac{1}{N}\sum_{l=1}^N H(Y_l | X_{1,l} X_{2,l} X_{3,l}).
\end{align*}
Let $P(X_{32,l}\neq X_{12,l}\oplus X_{12,l})=q_l$, for $l\in [1:N]$. Denote $\bar{q}_l=1-q_l$. We can show that,
\begin{align*}
H(Y_l | X_{1,l} X_{2,l} X_{3,l})=(1+2\bar{q}_l)h_b(\delta)+2q_l.
\end{align*}
We use the above argument, and the last inequality in (\ref{eq: last bound}) to give the following bound
\begin{align*}\nonumber
R_1+R_2+R_3 &\leq 3-\frac{1}{N}\sum_{l=1}^N [(1+2\bar{q}_l)h_b(\delta)+2q_l]+o(\epsilon)\\\label{eq: upper bound on I_N}
&= 3- 3 h_b(\delta)+\frac{1}{N}2(1-h_b(\delta))\sum_{l=1}^N q_l+o(\epsilon)
\end{align*}
By assumption $R_1+R_2+R_3  \geq 3(1-h_b(\delta)-\epsilon).$ Therefore, using the above bound we obtain,
\begin{align*}
\frac{3 \epsilon+o(\epsilon)}{2(1-h_b(\delta))} &\geq \frac{1}{N}\sum_{l=1}^N q_l \stackrel{(a)}{\geq }\frac{1}{N}\sum_{l\in \mathcal{I}^N_c} q_l,
\end{align*}
where $(a)$ holds, because we remove the summation over all $l \notin \mathcal{I}^N_c$. We defined  $\mathcal{I}^N_c$ as in the statement of this Lemma. Note that if $l \in \mathcal{I}^N_c$, then $q_l\geq c$. Finally, we obtain 
\begin{align*}
\frac{|\mathcal{I}^N_c|}{N} \leq \frac{3 \epsilon+o(\epsilon)}{2 c (1-h_b(\delta))}
\end{align*}
\end{proof}

\subsection{Proof of Lemma \ref{lem: structure in the code}}\label{seq: lem 3}
\begin{proof}
Let $\mathcal{I}_c^N$ be as in Lemma \ref{lem: x_2+x_1 needs to be decoded}. The average probability of error for decoding $X_{12}^N\oplus X_{22}^N$ is bounded as 
\begin{align*}
\bar{P}_e&=\frac{1}{N}\sum_{l=1}^N P(X_{32,l}\neq X_{12,l}\oplus X_{22,l})\\
&=\frac{1}{N}\sum_{l\in \mathcal{I}_c^N} P(X_{32,l}\neq X_{12,l}\oplus X_{22,l})+\frac{1}{N}\sum_{l\notin \mathcal{L}_c^N} P(X_{32,l}\neq X_{12,l}\oplus X_{22,l})\\
&\leq \frac{|\mathcal{I}_c^N|}{N}+c(1-\frac{|\mathcal{I}_c^N|}{N})\\
&=(1-c)\frac{|\mathcal{I}_c^N|}{N} +c\\
&\leq (1-c)\frac{\eta(\epsilon)}{2c(1-h(\delta))}+c 
\end{align*}
As a result as $\epsilon\rightarrow 0$, then $\bar{P}_e\rightarrow c$. Since $c>0$ is arbitrary, $\bar{P}_e$ can be made arbitrary small. Hence, for any $\epsilon'>0$, and there exist $\epsilon>0$ and large enough $N$ such that $\bar{P}_e < \epsilon'$. Note that $X^N_{32}$ is a function of $M_3, Y_1^N, Y_{12}^N$ and $Y_{22}^N$. Next we argue that to get $\bar{P}_e < \epsilon'$, it is enough for $X_{32}^N$ to be a function of $M_3, Y_1^N$.  More precisely, given $X_{32, l}$, the random variables $Y_{12,l}$ and $Y_{22,l}$ are independent of $X_{12, l}\oplus X_{22, l}$. To see this, we need to consider two cases.  If $X_{32, l}=X_{12, l}\oplus X_{22, l}$ then the argument follows trivially. Otherwise, $Y_{12,l}=X_{12,l}\oplus N_{1/2}$, where $N_{1/2}\sim Ber(1/2)$, and it is independent of $X_{12,l}$. Hence in this case,  $Y_{12,l}$ is independent of $X_{12,l}$. Similarly, $Y_{22,l}$ is independent of $X_{22,l}$. 

By subtracting $X_{31}^N$ from $Y_1^N$, we get $Z^N := X_{11}^N \oplus X_{21}^N\oplus N_{\delta}^N$. Next, we argue that the third encoder uses $Z^N$ to decode $X_{12}^N\oplus X_{22}^N$. Since $M_3$ is independent of $M_1$ and $M_2$, it is independent of $X_{1j}^N, X_{j2}^N$ for $j=1,2$. Therefore $Z^N$ is independent of $M_3$. Hence, $X_{32}^N$ is function of $Z^N$. Intuitively, we convert the problem of decoding $X_{11}^N \oplus X_{21}^N$ to a point to point channel coding problem. The channel in this case is a binary additive channel with noise $N_\delta \sim Ber(\delta)$. In this channel coding problem the codebook at the encoder is $\mathcal{C}_{12}\oplus \mathcal{C}_{22}$.  The capacity of this channel equals $1-h_b(\delta)$. Since the average probability of error is small,  we can use the generalized Fano's inequality to bound the rate of the encoder. As a result, it can be shown that  
\begin{equation}\label{eqe:bound on sum codebook}
\frac{1}{N}\log_2||\mathcal{C}_{12}\oplus \mathcal{C}_{22}|| \leq 1-h_b(\delta)+ \eta(\epsilon),
\end{equation}
where $\eta(\epsilon)\rightarrow 0$ as $\epsilon \rightarrow 0$. 
\begin{lem}
The following bound holds 
\begin{align} \label{eqe: bound on C_12 and C_22}
\frac{1}{N}\log_2||\mathcal{C}_{j2}|| \geq 1-h_b(\delta)- \gamma_j(\epsilon),
\end{align}

 where $j=1,2$ and $\gamma_j(\epsilon)\rightarrow 0$ as $\epsilon \rightarrow 0$.
\end{lem}
\begin{proof}[Outline of the proof]
First, we  show that the decoder must decode $M_3$ from $Y_1^N$. We argued in the above that $X_{32}^N$ is independent of $M_3$. Hence, the message $M_3$ is encoded only to $X_{31}^N$. Since $X_{31}^N$ is sent though the first channel in Example 1, the decoder must decode $M_3$ from  $Y_1^N$. Next, we argue that the receiver must decode $M_1$ and $M_2$ from $Y_{21}^N$ and $Y_{22}^N$, respectively. Note that the rate of the third encoder is $1-h_b(\delta)$, which equals to the capacity of the first channel given $X_{11}^N \oplus X_{21}^N$. Therefore, the decoder can decode $M_3$, if it has $X_{11}^N \oplus X_{21}^N$. Hence, the decoder must reconstruct $X_{11}^N \oplus X_{21}^N$ from the second channel. It can be shown that this is possible, if the decoder can decode $M_1 $ and $M_2$ from the second channel. As a result, from Fano's inequality, the bounds in the Claim hold.  
\end{proof}

Finally, using \eqref{eqe:bound on sum codebook} and (\ref{eqe: bound on C_12 and C_22}) we get 
\begin{align*}
0 \leq \frac{1}{N}\log_2||\mathcal{C}_{12}\oplus \mathcal{C}_{22}||-\frac{1}{N}\log_2||\mathcal{C}_{j2}|| \leq \eta(\epsilon)+\gamma_j(\epsilon), \quad j=1,2.
\end{align*}
This completes the proof.
\end{proof}

\bibliographystyle{IEEEtran}


\end{document}